\documentclass[a4paper,UKenglish,cleveref, autoref, thm-restate, pdfa]{lipics-v2021}

\hideLIPIcs  

\usepackage{framed}
\usepackage{tikz}
\usetikzlibrary{arrows,arrows.meta,automata,positioning,matrix,calc,petri,backgrounds,shapes,decorations.pathreplacing,decorations.markings,fit}
\usepackage{stmaryrd}
\usepackage{bm}  
\usepackage{algpseudocode}
  \algrenewcommand\alglinenumber[1]{\textcolor{gray}{\scriptsize \tt #1:}}
  \algtext*{EndIf}
  \algtext*{EndFor}
  \algtext*{EndWhile}
  \algtext*{EndFunction}

\usepackage{xspace}

\nolinenumbers

\newif\iffinal
\finaltrue
\iffinal
\else
\fi

\iffinal
\else
\makeatletter
\def\ps@headings{%
  \let\@oddhead\@empty
  \let\@evenhead\@empty
  \def\@oddfoot{\hfil\tiny\thepage\hfil}%
  \def\@evenfoot{\hfil\tiny\thepage\hfil}%
  \let\@mkboth\markboth
  \let\sectionmark\@gobble
  \let\subsectionmark\@gobble
}
\makeatother
\fi

\sethlcolor{pink}

\newcommand{\trans}[1]{\stackrel{{#1}}{\longrightarrow}}

\newcommand\ak[1]{}
\newcommand{\LTLX}{\mathsf{X}}
\newcommand{\LTLG}{\mathsf{G}}
\newcommand{\LTLF}{\mathsf{F}}
\newcommand{\LTLFG}{\mathsf{FG}}
\newcommand{\LTLGF}{\mathsf{GF}}
\newcommand{\LTLU}{\,\mathcal{U}}

\newcommand{\impl}{\to}

\newcommand\SuffLang{{\sf Suffix}}

\newcommand\bbN{\mathbb{N}}

\newcommand\mc\mathcal
\newcommand\x\times

\newcommand\pto\rightharpoonup
\newcommand\ellmo{{\ell-1}}
\newcommand\ellpo{{\ell+1}}
\newcommand\Vinit{V^\mathit{init}}

\newcommand\projDl{{\mid \mc F^\ell}}
\newcommand\projDll{{\mid \mc F^{\ell+1}}}
\newcommand\lPlus{{\ell+1}}
\newcommand\Dl{\mc F^\ell}
\newcommand\Fl{\mc F^\ell}
\newcommand\Fellmo{\mc F^\ellmo}
\newcommand\Fellpo{\mc F^\ellpo}
\newcommand\Dll{\mc F^{\ell+1}}
\newcommand\Dlp{\mc F^{\ell'}}%
\newcommand\projDlp{{\mid \mc F^{\ell'}}}%

\newcommand\Gl{{\mc G^\ell}}

\newcommand\mcG{\mc G}
\newcommand\tmcG{{\tilde{\mc G}}}

\newcommand\DFA{DfA\xspace}
\newcommand\DFAs{DfAs\xspace}
\newcommand\NFA{NfA\xspace}
\newcommand\NFAs{NfAs\xspace}

\newcommand\sem[1]{L(#1)}
\newcommand\simFor[1]{\sim_{#1}}

\newcommand\li{\begin{itemize}}
\newcommand\il{\end{itemize}}
\renewcommand{\-}\item
\newcommand\lo{\begin{enumerate}}
\newcommand\ol{\end{enumerate}}

\newcommand\parbf[1]{\smallskip\noindent\textbf{#1}}  
\newcommand\parit[1]{\smallskip\noindent\textit{#1}}  

\newcommand\oneexp{2^{O(n)}}
\newcommand\twoexp{2^{2^{O(n)}}}
\newcommand\oneexpfi{2^{O(|\varphi|)}}
\newcommand\twoexpfi{2^{2^{O(|\varphi|)}}}

\newcommand\Iff\Leftrightarrow
\newcommand\Impl\Rightarrow

\let\emptyset\varnothing

\newtheorem{fact}[theorem]{Fact}
\newtheorem*{fact*}{Fact}

\newtheorem*{assumption*}{Assumption}

\AtBeginDocument{%
  \setlength{\abovedisplayskip}{4pt plus 1pt minus 1pt}%
  \setlength{\belowdisplayskip}{4pt plus 1pt minus 1pt}%
  \setlength{\abovedisplayshortskip}{2pt plus 1pt minus 1pt}%
  \setlength{\belowdisplayshortskip}{2pt plus 1pt minus 1pt}%
  \setlength{\jot}{2pt}
}

\title{A Naturally-Colored Translation from LTL\\ to Parity and COCOA}
\titlerunning{A Naturally-Colored Translation from LTL to Parity and COCOA}

\author{R\"udiger Ehlers}{Clausthal University of Technology, Germany \and \url{http://www.ruediger-ehlers.de}}{ruediger.ehlers@tu-clausthal.de}{https://orcid.org/0000-0002-8315-1431}{Funded by Volkswagen Foundation within its Momentum framework under project no. 9C283}

\author{Ayrat Khalimov}{Clausthal University of Technology, Germany}{ayrat.khalimov@tu-clausthal.de}{https://orcid.org/0000-0001-8277-5501}{Funded by Volkswagen Foundation within its Momentum framework under project no. 9C283} 

\authorrunning{R.Ehlers and A.Khalimov} 

\Copyright{R\"udiger Ehlers and Ayrat Khalimov} 



\category{} 

\relatedversion{} 

\begin{document}

\setlength{\mathindent}{10pt}   



\EventEditors{Claudia Faggian and Joost-Pieter Katoen}
\EventNoEds{2}
\EventLongTitle{41st Annual Symposium on Logic in Computer Science (LICS 2026)}
\EventShortTitle{LICS 2026}
\EventAcronym{LICS}
\EventYear{2026}
\EventDate{July 20--23, 2026}
\EventLocation{Lisbon, Portugal}
\EventLogo{}
\SeriesVolume{380}
\ArticleNo{71}

\maketitle

\begin{abstract}
  Chains of co-Büchi automata (COCOA) have recently been introduced as a new canonical representation of omega-regular languages.
  The co-Büchi automata in a chain assign each omega-word its natural color, which depends only on the language itself and not on the chosen automaton representation.
  Automata in such a chain can be minimized in polynomial time and are good-for-games,
  making this representation attractive for verification and reactive synthesis.
  However, in these applications, specifications are usually given in linear temporal logic (LTL).
  To make COCOA useful, an LTL specification must first be translated into the chain of automata.
  The only translation currently known proceeds via deterministic parity automata (LTL$\,{\to}\,$DPA$\,{\to}\,$COCOA),
  where the first step ignores natural colors and requires involved constructions due to Safra or Esparza et al.
  This raises the question of whether, by exploiting the definition of the natural color of words,
  one can avoid such constructions and obtain a direct translation from LTL to COCOA.

  In this paper, we present a simple yet optimal translation from LTL to COCOA,
  as well as a variant that translates LTL into DPA.
  The translation represents a new path from LTL to DPA and exploits the definition of natural colors.
  It relies on standard operations on weak alternating automata,
  the Miyano-Hayashi breakpoint construction,
  the subset construction, and simple graph algorithms.
  Starting from weak alternating automata,
  the procedure also applies to specifications in linear dynamic logic.
  The procedure runs in asymptotically optimal doubly exponential time and produces automata of asymptotically optimal size.
\end{abstract}

\section{Introduction}
Canonical forms play an important role in computer science.
A typical canonical form is born from an insight about the nature of objects it describes,
leading to the development of efficient algorithms manipulating the objects.
In the context of finite words, the Myhill-Nerode theorem about regular languages brought to life a canonical form of finite automata, and led to efficient automata minimization and learning algorithms.

Omega-regular languages do not possess a simple Myhill-Nerode-like characterization.
This led to the flourishing of canonical forms based on representing a language as a family of automata over finite words~\cite{10.1007/3-540-58027-1_27,MALER199793,ANGLUIN201657,10.1007/978-3-031-45329-8_3}.
These models are based on the insight that an omega-regular language can be characterized by the language of its ultimately-periodic words (those of the form $u v ^\omega$).
As a result, these models operate on finite words of the form $u \cdot v$,
which makes them a perfect fit for automata-learning tasks but ill-suited for classical problems of model checking and synthesis.
An alternative canonical representation that does not assume words to be ultimately periodic, and hence is more conventional, was recently introduced by Ehlers and Schewe in~\cite{DBLP:conf/fsttcs/EhlersS22}.
It is based on the notion of \emph{natural colors}.

Natural colors assign colors to words based solely on an omega-regular language and independently of its representation.
This allows an omega-regular language $L$ to be represented as a unique chain of co-B\"uchi languages $L^1 \supset \ldots \supset L^k$,
where $L^i$ is the set of words whose natural color is greater or equal $i$,
and $k$ is the smallest number of colors in a parity automaton recognizing $L$.
By definition, a word belongs to $L$ if and only if its natural color wrt.\ $L$ is even.
Each co-B\"uchi language $L^i$ can be represented by a canonical (and minimal) transition-based history-deterministic co-B\"uchi automaton~\cite{DBLP:journals/lmcs/RadiK22}.
Thus, every omega-regular language $L$ has a unique representation as a chain of such co-B\"uchi automata (\emph{COCOA}) $\mc A^1,\ldots,\mc A^k$.

The two distinguishing features of COCOA ---
(1) the chain representation with clearly defined semantics, and
(2) the minimality and history-determinism of automaton at each level ---
can potentially be exploited to design better verification and synthesis algorithms. For instance,
COCOA were recently used to solve a reactive synthesis problem that previous tools were unable to handle~\cite{EK24}.

However, verification and synthesis problems are usually defined over system behavior descriptions written in linear temporal logic (LTL).
Therefore, to make COCOA useful to verification and synthesis applications,
one needs a procedure for translating LTL to COCOA.
Currently, the only known way to obtain COCOA from LTL formulas is by going through deterministic parity word  automata (DPA):
LTL $\to$ DPA $\to$ COCOA~\cite{DBLP:conf/fsttcs/EhlersS22}.
Although the second step is relatively simple, the first step represents a challenge.
For instance,
one can first translate the LTL property to a nondeterministic Büchi automaton and then apply Safra's determinization construction~\cite{DBLP:conf/focs/Safra88}.
However, it is well-known that implementing Safra's construction efficiently is challenging.
Alternatively, one can use the LTL-to-DRA translation of Esparza et al.~\cite{DBLP:journals/jacm/EsparzaKS20} that aims to generate automata with a more regular structure and is thus more amenable to efficient implementations.
Nevertheless, in both cases, the overall translation neglects the insight of natural colors in the main challenging step.
This observation raises the question ``\emph{Is there a direct translation from LTL to COCOA?}''.

The second question tackled in this paper concerns the problem of translating from LTL to DPA.
Modern approaches to reactive synthesis and probabilistic verification often reduce to solving parity games~\cite{renkin.23.fmsd,strix,KNP11}, and hence translate from LTL to DPA as an integral part.
Experience from reactive synthesis~\cite{SyntComp} shows that current approaches to this translation still constitute a bottleneck, indicating a need for further research.
We hence ask whether natural colors can enable a conceptually simpler translation LTL$\,{\to}\,$DPA,
which in turn could enable efficient implementations.

In this paper,
we describe a simple yet optimal translation from LTL to COCOA,
and a variant that translates into DPA.
Our translation procedure starts from weak alternating automata,
hence it can also translate formulas in linear dynamic logic (LDL)~\cite{DBLP:conf/ijcai/GiacomoV13,DBLP:journals/iandc/FaymonvilleZ17} which subsumes LTL.
The procedure uses simple operations on weak alternating automata,
Miyano-Hayashi's breakpoint construction~\cite{MHConstruction},
the subset construction, and easy graph algorithms.
The translation works in asymptotically optimal doubly exponential time.
The procedure constitutes a novel path for translating from LTL/LDL to COCOA and DPA,
avoiding detours to Safra~\cite{DBLP:conf/focs/Safra88} and Esparza et al~\cite{DBLP:journals/jacm/EsparzaKS20}.
We now describe ideas behind the new approach.

\paragraph*{LDL$\,\to\,$COCOA}
We translate a given LDL formula into a weak alternating automaton $\mc A$ and its dual $\bar {\mc A}$ using standard constructions~\cite{MSS88,DBLP:journals/iandc/FaymonvilleZ17}.
We then apply the Miyano-Hayashi breakpoint construction~\cite{MHConstruction}
to eliminate universal transitions,
and obtain two nondeterministic B\"uchi automata,
called obligation graphs:
$\mc G_\mc A$ with language $L(\mc A)$ and $\mc G_{\bar{\mc A}}$ with language $\overline{L(\mc A)}$.
Using simple operations on alternating automata and Miyano-Hayashi's construction,
we construct a machine that tracks the current suffix language wrt.\ $L(\mc A)$.
States of this suffix-language-tracking machine (SLTM) are labelled by the vertices of obligation graphs $\mc G_\mc A$ and $\mc G_{\bar {\mc A}}$ that are reachable by prefixes whose suffix language matches the one tracked by the current SLTM state.
The SLTM and the special labelling enables the construction of the first level $\mc A^1$ of the COCOA for $L(\mc A)$,
using the product operation on the SLTM and $\mc G_{\bar {\mc A}}$,
the subset construction, and simple graph analysis.
This construction relies on floating automata~\cite{ehlers2025rerailingautomata},
which are looping (safety) automata synchronized with the SLTM and allow a ``run jump''.
They are a convenient replacement for co-B\"uchi automata,
tailored to represent COCOA-level languages,
and can be translated back and forth in polynomial time.
Once the level-1 floating automaton is constructed,
the construction at level-2 uses a similar procedure except that it builds the product not with the SLTM but with the level-1 floating automaton.
Each such floating automaton can be constructed in time $\twoexp$\!, where $n$ is the number of states in $\mc A$,
and translated in polynomial time into a history-deterministic co-B\"uchi automaton,
then minimized-canonized using the effective procedure of~\cite{DBLP:journals/lmcs/RadiK22}.
Since the number of levels of COCOA is at most $\oneexp$,
this yields an overall $\twoexp$-time procedure for translating from LDL to COCOA.

\paragraph*{LDL$\,\to\,$DPA}
The procedure to translate LDL into DPA reuses the floating automata constructed by the LDL-to-COCOA procedure.
The DPA simulates these floating automata in parallel and \emph{tries} to follow the automaton on the deepest level possible.
When the DPA does not manage to reach the deepest possible floating automaton,
it ensures that the parity of the level it actually reaches is preserved.
Each DPA state is a state of some floating automaton, and each DPA transition is colored with the index of the deepest-level floating automaton
that can take that transition assuming it starts from the state of the DPA.
Every floating automaton constructed by the LDL-to-COCOA procedure has $\twoexp$ states,
there are at most $\oneexp$ levels in the COCOA,
so the constructed DPA has at most $\twoexp$ states.

\paragraph*{LDL$\,\to\,$LDBA}
Limit-deterministic B\"uchi automata (LDBA) were introduced as a compromise between fully deterministic and general nondeterministic automata
for solving the qualitative probabilistic model-checking problem~\cite{10.1145/210332.210339}.
An LDBA is a special kind of nondeterministic B\"uchi automaton
whose state set is partitioned into two disjoint parts.
Every run starts in the first part, which is nondeterministic and has no accepting states.
From the first part, every accepting run must eventually transit into the second part;
there are no transitions from the second to the first part.
The second part is deterministic and contains accepting states.
Our procedure for translating LDL into LDBA
uses deterministic floating automata $\mc F^1\!,\ldots,\mc F^k$ constructed during the translation LDL\,$\to$\,COCOA.
By that construction,
a word belongs to the language if and only if there exists an even number $i \in \{0,\ldots,k\}$
such that the word is accepted by $\mc F^i$ and rejected by $\mc F^{i+1}$\!,
where, by convention, $\mc F^0$ accepts everything and $\mc F^{k+1}$ accepts nothing.
Our procedure constructs an LDBA whose first part is the SLTM itself.
From the SLTM, a run can jump into a deterministic component that always follows the transitions of $\mc F^i$
but infinitely often falls out of $\mc F^{i+1}$.
We can construct such a deterministic component with a B\"uchi acceptance condition thanks to
simple safety acceptance condition of floating automata.

\subparagraph*{Paper structure.}
Section~\ref{sec:prelims} introduces known concepts
including weak alternating automata and COCOA, and states a few useful facts about them.
Section~\ref{sec:WAA-to-COCOA} first defines SLTMs and floating automata,
then describes the main translation procedure LDL\,$\to$\,COCOA.
Section~\ref{sec:LDL-to-DPA} describes our translation LDL\,$\to$\,DPA,
and Section~\ref{sec:LDL-to-LDBA} describes the translation LDL\,$\to$\,LDBA.
Section~\ref{sec:related-works} provides a detailed comparison of our work with \cite{DBLP:journals/jacm/EsparzaKS20},
and Section~\ref{sec:conclusion} concludes.

\section{Preliminaries}\label{sec:prelims}
\ak{todo: add ``at most'' to complexity claims}

Let $\bbN = \{0,1,\ldots\}$ denote the set of natural numbers including $0$.
We assume that the reader is familiar with LTL;
its formulas are built over atomic propositions
using the standard Boolean operators and the temporal operators $\LTLU$ and $\LTLX$,
as well as the derived operators $\LTLG$ and $\LTLF$.
Familiarity with LDL is not required.
Satisfaction of LDL and LTL formulas is defined wrt.\ infinite words in $\Sigma^\omega$,
where $\Sigma = 2^{\mathit{AP}}$ is the alphabet over atomic propositions $\mathit{AP}$.
The set of models of an LTL/LDL formula $\varphi$ is denoted by $L(\varphi)$.

\subsection{Weak alternating automata}

A useful automaton model for representing LDL/LTL specifications on a more technical level are \emph{weak alternating automata}.

Given a set of propositions $P$,
let $\mc B^+(P)$ denote the set of positive Boolean formulas built from $P$ using conjunction and disjunction.

A \emph{weak alternating B\"uchi automaton} is a tuple $\mc A = (\Sigma,Q,q_0,\delta,F)$
where
$\Sigma$ is an alphabet,
$Q$ is a finite set of states,
$\delta : Q \x \Sigma \to \mc B^+(Q)$ is a transition function,
$q_0$ is an initial state, and
$F \subseteq Q$ is a set of accepting states.
The prefix ``weak'' in the name means that the automaton has a total preorder $\gtrsim$ on states $Q$
which satisfies:
(i) for every $q,x,q'$ where $q'$ appears in $\delta(q,x)$: $q \gtrsim q'$, and
(ii) states equivalent wrt.\ $\gtrsim$ are either all accepting or all rejecting (not in $F$).
Intuitively,
an alternating automaton is weak if in every of its strongly connected components (when ignoring edge labels),
either all states are accepting, or all states are rejecting.

We do not define the semantics of alternating automata, as it is not needed in the paper.
The following result follows from the proof of Thm.2 in \cite{DBLP:journals/iandc/FaymonvilleZ17}.

\newcommand\factLDLtoWeakStatement{%
  Every LDL (and hence LTL) formula $\psi$ can be translated
  into a weak alternating B\"uchi automaton recognizing the models of $\psi$
  with $O(|\psi|)$ states and a transition function of size $2^{O(|\psi|)}$\! in time $2^{\textit{poly}(|\psi|)}$\!.
  }

\begin{fact}[\cite{DBLP:journals/iandc/FaymonvilleZ17}]
  \label{fact:LDL-to-Weak}
  \factLDLtoWeakStatement
\end{fact}

This paper assumes that the transition functions of alternating automata
do not use pathologically large Boolean formulas,
and the alphabet size is at most linearly exponential in the number of automaton states;
this holds for automata obtained by translation from LDL.

\begin{assumption*}
  We consider only weak alternating automata with $n$ states over alphabets satisfying the following condition:
  there exists a function $B(n) \in \oneexp$, depending only on $n$, such that
  the size of the alphabet and the size of each Boolean formula in the transition function are bounded by $B(n)$.
\end{assumption*}

\begin{lemma}\label{lem:A-trans-func-bound}
  Under the above assumptions,
  the size of the transition function of a weak alternating automaton with $n$ states is at most $\oneexp$.
\end{lemma}

\subsection{Automata DPA, NBA, NCA, HD-NCA}
We explicitly define nondeterministic automata with transition- and state-based acceptance
to keep the definition of alternating automata relatively simple.

A \emph{nondeterministic parity automaton} is a tuple $\mc A = (\Sigma, Q, Q_0, \delta, \alpha)$
where $\Sigma$ is an alphabet,
$Q$ is a set of states,
$Q_0 \subseteq Q$ are initial states, and
$\delta \subseteq Q \x \Sigma \x Q$ is a transition relation.
In automata with \emph{transition-based} parity acceptance,
the function $\alpha: \delta \to \bbN$ maps transitions to colors,
while with \emph{state-based} parity acceptance, $\alpha: Q \to \bbN$ maps states to colors.

A \emph{path} of $\mc A$ is a (finite or infinite) sequence of transitions
$\pi = (q_1,x_1,q_2)(q_2,x_2,q_3)\ldots$
where $(q_i, x_i, q_{i+1}) \in \delta$ for all $i \in \{1,\ldots,|\pi|\}$
when $\pi$ is finite and for all $i \geq 1$ when $\pi$ is infinite.
A \emph{run} of $\mc A$ on an infinite word $x_0 x_1 \ldots$
is an infinite path $(q_0,x_0,q_1)(q_1,x_1,q_2)\ldots$
that starts in an initial state $q_0 \in Q_0$;
a word might have one run, several runs, or no runs at all.
The \emph{color sequence} of a run $(q_0,x_0,q_1)(q_1,x_1,q_2)\ldots$
is the sequence $c_1 c_2 \ldots$,
where, for all $i$,
$c_i = \alpha(q_i)$ in the case of state-based automata
and $c_i = \alpha(q_i,x_i,q_{i+1})$ for transition-based automata.
A run is \emph{accepting} if the minimal color appearing infinitely often in the color sequence is even
(min-even acceptance).
A word is \emph{accepted} if it has an accepting run.
The \emph{language} of $\mc A$ is the set of words accepted by $\mc A$.

\emph{B\"uchi automata} are defined by restricting the image of $\alpha$ to the two colors $\{0,1\}$;
in this case, instead of $\alpha$, we use $F\subseteq \delta$ (or $F\subseteq Q$),
and call transitions (or states) in $F$ \emph{accepting}.
Then, a run is accepting if and only if an accepting transition (or state) is visited infinitely often.
Similarly,
in \emph{co-B\"uchi} automata, the function $\alpha$ is restricted to the two colors $\{1,2\}$;
then, a run is accepting if and only if eventually only accepting transitions/states are visited.

An automaton is \emph{deterministic}
if $|Q_0|=1$ and for every $q \in Q, x\in \Sigma$: $|\{(q,x,q'): (q,x,q') \in \delta\}| = 1$.

A nondeterministic parity automaton $\mc A$ is \emph{history-deterministic} (also called \emph{good-for-games})
if there exists a strategy function $\lambda : \Sigma^* \to Q$
such that for every word $x_0 x_1 \ldots$ accepted by $\mc A$,
the sequence $(q_0,x_0,q_1) (q_1,x_1,q_2)\ldots$,
where $q_i = \lambda(x_0 \ldots x_{i-1})$ for all $i \geq 0$,
is an accepting run of $\mc A$ on the word.

An automaton $\mc A = (\Sigma,Q,Q_0,\delta,\alpha)$
can be associated with the following directed graph $(Q,\to)$:
$q\to q'$ if and only if
there is $x$ such that $q' \in \delta(q,x)$.
\emph{An SCC (strongly connected component) of $\mc A$} is an SCC of the associated graph $(Q,\to)$.
An SCC is \emph{maximal}
if it cannot be extended without losing the property of being strongly connected.

For automata types, we use standard abbreviations:
DPA, NBA, NCA, and we add ``t'' to highlight transition-based acceptance (e.g.\ HD-tNCA).

\subsection{Canonical co-B\"uchi automata}

Recall that deterministic and nondeterministic co-B\"uchi automata have the same expressive power.
A language is \emph{co-B\"uchi} if there is an NCA recognizing the language.
Size-wise,
HD-tNCAs are exponentially more succinct than tDCAs~\cite{kuperberg2015determinisation}.
HD-tNCAs also admit a canonical form,
introduced in~\cite{DBLP:journals/lmcs/RadiK22}.
Automata in this canonical form are the smallest in the class of HD-tNCAs.

\begin{fact}[\cite{DBLP:journals/lmcs/RadiK22}]\label{fact:min-can-HD-tNCA}
  Any HD-tNCA can be transformed into the canonical form in polynomial time.%
\end{fact}

\subsection{Canonical chain representation (COCOA)}

We start by defining a few helpful notions.
Let $\SuffLang : 2^{\Sigma^\omega} \!\!\times \Sigma^* \to 2^{\Sigma^\omega}$
be a function mapping a language $L \subseteq \Sigma^\omega$ and a finite word $p \in \Sigma^*$
to $\{ w \in \Sigma^\omega \mid p\cdot w \in L\}$,
i.e., the \emph{suffix language} of $L$ for $p$.
For brevity, we may write $L^\textit{suff}(p)$ instead of $\SuffLang(L,p)$.
The set of all suffix languages of $L$ is denoted by $L^\textit{suff}$.
We write $p_1 \simFor{L} p_2$ when $L^\mathit{suff}(p_1) = L^\mathit{suff}(p_2)$;
the relation $\sim_L$ is called \emph{right-congruence relation}.
A finite word $u \in \Sigma^*$ \emph{preserves} the suffix language of $p$ wrt.\ $L$
if $L^\mathit{suff}(p \cdot u) = L^\mathit{suff}(p)$.
For an infinite word $w = w_0 w_1 \ldots$ and an infinite set $J \subseteq \bbN$ of positions $j_1<j_2<\ldots$,
we call a word $w'$ the result of a \emph{suffix-language-preserving injection} into $w$ at $J$
if it is of the form $w' = w_0 ... w_{j_1-1} u_1 w_{j_1} ... w_{j_2-1}u_2w_{j_2}\ldots$,
where, for every $i$, the finite word $u_i$ preserves the suffix language of $w_0 \ldots w_{j_i-1}$ wrt.\ $L$.
Note that, because $J$ is infinite,
a suffix-language-preserving injection into a word $w \in L$ may result in $w' \not\in L$, and vice versa.
For instance,
consider the language $\sem{\LTLFG a}$ over $\Sigma=\{a,\bar a\}$ and its word $a^\omega$:
the word $(a\bar a)^\omega$ is the result of a suffix-language-preserving injection into $a^\omega$ at positions $J = \{1,2,3,\ldots\}$
and $(a\bar a)^\omega \not\in \sem{\LTLFG a}$.
We write $w' \in \text{SLII}^L_{(J,w)}$, or $\text{SLII}^L_{(J,w)}\ w'$,
when $w'$ is the result of a suffix-language-preserving injection into $w$ at positions $J$,
and we may omit $L$, $J$, $w$ when they are clear from the context.

The \emph{canonical chain-of-co-B\"uchi-automata representation (\hspace{-0.05mm}\textbf{COCOA}\hspace{-0.3mm})} of an omega-regular language $L$
is the unique sequence of canonical HD-tNCAs $(\mc A^1, \ldots, \mc A^k)$
where each automaton $\mc A^i$ is defined through its language $L^i$ inductively:
let $L^0 = \Sigma^\omega$, then for all $i$:
$$
  L^i =
     \big\{
       w \in L^{i-1} \mid
       \forall J.\,\exists \text{SLII}^L_{(J,w)} w'{:~}
       w' \in L^{i-1} \cap \left[\!\!\!\!\begin{aligned} \renewcommand{\arraystretch}{0.8} \begin{array}{c} \bar{L} \textit{\small~for odd $i$ } \\ L \textit{\small~for even $i$} \end{array}\end{aligned}\!\!\!\right]\!
     \big\},
$$
and $k \geq 0$ is the largest number such that $L^k$ is non-empty
($k=0$ means that the automaton sequence is empty).
The number $k+1$ defines the number of colors of the COCOA.
The \emph{natural color} of a word $w$ wrt.\ $L$
is the maximal index of an automaton in the chain accepting the word,
it is $0$ when $\mc A^1$ rejects $w$ or the chain is empty.
Ehlers and Schewe~\cite{DBLP:conf/fsttcs/EhlersS22} showed
that every omega-regular language has a unique COCOA, and
that a word $w$ belongs to $L$ if and only if the natural color of $w$ wrt.\ $L$ is even\label{page:def:cocoa} (max-even acceptance).

We elaborate on the definition.
First,
observe that every COCOA $(\mc A^1,\ldots,\mc A^k)$ forms a chain of shrinking languages:
$L^1 \supset \ldots \supset L^k$,
where $L^i=L(\mc A^i)$ for every $i$.
Second,
each language $L^{\text{even } i}$ is obtained
by removing from $L^{i-1}$ every $\bar L$-word
that, for some $J$, cannot be flipped into an $L$-word of $L^{i-1}$ via the SLII operation.
Dually,
each $L^{\text{odd } i}$ is obtained
by removing from $L^{i-1}$ every $L$-word
that, for some $J$, cannot be flipped into an $\bar L$-word of $L^{i-1}$.
Thus, even levels remove certain negative words, whereas odd levels remove certain positive words.

\begin{example}
  In the following examples,
  we do not explicitly define the individual automata of COCOA but rather their languages.
  \emph{Example 1}: Consider the safety language $L(\LTLG a)$ over the alphabet $\Sigma = 2^{\{a\}}$.
  Its COCOA consists of a single automaton $\mc A^1$ with $L(\mc A^1) = L(\LTLF \neg a)$.
  Indeed, for every $w\in L(\mc A^1)$ and suffix index set $J\subseteq \bbN$,
  we take the SLII $w' = w$ to exhibit $w' \in \bar L$.
  Note that for the original language $L(\LTLG a)$,
  only insertions of the form $a^*$ are suffix-language-preserving,
  implying that no other word can be added to $L(\mc A^1) = L(\LTLF \neg a)$.
  Every word $w \in L(\LTLG a)$ has natural color $0$.
  \emph{Example 2}: Consider the language $L(\LTLFG a)$.
  Its COCOA $(\mc A^1,\mc A^2)$ has $L(\mc A^1) = \Sigma^\omega$ and $L(\mc A^2) = L(\LTLFG a)$, justified as follows.
  For any prefix $p$, the suffix language of $L(\LTLFG a)$ after reading $p$ equals $L(\LTLFG a)$,
  hence inserting anything preserves the suffix language.
  In particular, we can insert $\bar a$, yielding the SLII $w' \in \overline{L(\LTLFG a)}$ regardless of $w$,
  thus $L(\mc A^1) = \Sigma^\omega$.
  The case of $\mc A^2$ is straightforward.
  \emph{Example 3}:
  Consider the language $L(\LTLGF a \to \LTLGF b)$ over $\Sigma = 2^{\{a,b\}}$:
  it has the COCOA $(\mc A^1,\mc A^2)$ with
  $L(\mc A^1) = L(\LTLFG\neg b)$ and $L(\mc A^2) = L(\LTLFG \neg a \land \LTLFG \neg b)$, justified as follows.
  The language $L(\mc A^1)$ must contain every word from $L(\LTLGF a \land \LTLFG \neg b)$ and
  all $L$-words that can be turned into an $\bar L$-word by the SLII operation for every given $J$.
  Since the original language is liveness LTL, inserting anything is suffix-language-preserving.
  Then is not hard to see that $L(\mc A^1) = L(\LTLFG \neg b)$.
  The language $L(\mc A^2)$ is the result of removing $\bar L$-words from $L(\mc A^1)$
  for which there exists $J$ such that no SLII at $J$ results in a word in $L\cap L(\mc A^1)$.
  We leave it as an exercise to show that $L(\mc A^2) = L(\LTLFG \neg a \land \LTLFG \neg b)$.
  \emph{Example 4}:
  Finally, the language $L(\LTLGF a \to (\LTLGF b \land \LTLFG c))$ over $\Sigma = 2^{\{a,b,c\}}$ has the COCOA $(\mc A^1,\mc A^2,\mc A^3,\mc A^4)$ with
  $L(\mc A^1) = \Sigma^\omega$,
  $L(\mc A^2) = L(\LTLFG \neg a \lor \LTLFG c)$,
  $L(\mc A^3) = L(\LTLFG c \land \LTLFG \neg b)$, and
  $L(\mc A^4) = L(\LTLFG c \land \LTLFG \neg b \land \LTLFG \neg a)$.
\end{example}

\newcommand\factCocoaSizeStatement{%
  Every weak alternating automaton $\mc A$ with $n$ states can be translated into
  the COCOA recognizing $L(A)$ with
  $2^{O(n)}$\! levels and $2^{2^{O(n)}}$\!\! states in total,
  in time $2^{2^{O(n)}}$\!\!.%
}

\begin{fact}[\cite{DBLP:conf/fsttcs/EhlersS22}]\label{fact:cocoa-nof-levels}\label{fact:cocoa-size}
  \factCocoaSizeStatement
\end{fact}
The proof idea is to translate a given weak alternating automaton into deterministic parity automaton using standard constructions~\cite{DBLP:conf/focs/Safra88,DBLP:journals/sttt/EsparzaKRS22},
then use the polynomial-time algorithm of \cite{DBLP:conf/fsttcs/EhlersS22} to translate the DPA into HD-tNCAs of the languages of individual COCOA levels.

The following fact states that the languages of individual COCOA levels are insensitive to congruent-prefix replacement.
\newcommand\factLevelPrefixIndependence{%
  Given an omega-regular language $L$,
  let $L^\ell$ be the language of the $\ell$th level of the COCOA for $L$.
  For any pair of prefixes $p_1,p_2$:\, $p_1 \simFor{L} p_2$ implies $p_1 \simFor{L^\ell} p_2$.
}
\begin{fact}[\cite{DBLP:conf/fsttcs/EhlersS22}]\label{fact:simL-implies-simLl}
  \factLevelPrefixIndependence
\end{fact}

\begin{proof}
  This fact can also be seen from Corollary~9 of \cite{DBLP:conf/fsttcs/EhlersS22}.
  Our proof is by induction on $\ell$;
  for $\ell=0$, where $L^0 = \Sigma^\omega$, the claim holds trivially.

  Suppose the claim holds for $L^\ellmo$, and consider $L^\ell$.
  Fix prefixes $p_1,p_2$ s.t.\ $p_1 \simFor{L} p_2$.
  By the inductive hypothesis,
  $p_1 \simFor{L^\ellmo} p_2$.
  We need to show that $p_1 \simFor{L^\ell} p_2$.
  Given a suffix word $s \in \Sigma^\omega$ such that $w_1=p_1 \cdot s \in L^\ell$,
  we need to prove that $w_2=p_2 \cdot s \in L^\ell$ (the other direction is symmetric).
  Thus,
  we show $(\dagger)$:
  $\forall J.\exists \text{SLII}^L_{(J,w_2)}w_2'\in \hat L\cap L^\ellmo$,
  where $\hat L$ is $L$ for even $\ell$ and $\bar L$ for odd $\ell$.

  Fix an index set $J$.
  Let $J'$ be the subset of $J$ starting from the earliest index greater than $|p_1|$.
  Since $p_1s \in L^\ell$,
  by definition there exists an SLII$^L_{(J'\!,w_1)} w_1' \in \hat L \cap L^\ellmo$.
  In this SLII $w'_1 = p_1 s'$, no injection occurs into $p_1$, all injections are done into $s$.
  We have:
  \li
  \- $p_1 s' \in \hat L ~\stackrel{p_1 \simFor{L} p_2}{\Longrightarrow}~ p_2 s' \in \hat L$.
  \- $p_1 s' \in L^\ellmo ~\stackrel{p_1 \simFor{L^\ellmo} p_2}{\Longrightarrow}~ p_2 s' \in L^\ellmo$.
  \il
  Therefore, $w'_2 = p_2s' \in \text{SLII}^L_{(J,w_2)}$, which proves $(\dagger)$.
  Hence $w_2=p_2s \in L^\ell$, implying $p_1 \simFor{L^\ell} p_2$.
\end{proof}

\section{Translating LDL into COCOA}\label{sec:WAA-to-COCOA}\label{sec:LDL-to-chain}

We want to build a checker in the form of an HD-tNCA
that, given a word $w$, checks the condition
$\forall J.\exists \text{SLII}^L_{(J,w)}~w': w' \in L^{\ell-1} \cap L(\mc G^\ell)$,
where $\Gl$ is a tNBA for $L$ when $\ell$ is even and for $\bar{L}$ when $\ell$ is odd.
Since tracking suffix languages is deeply ingrained in the definition of COCOA levels,
we start by defining suffix-language tracking machines.
We then introduce floating automata,
an automaton formalism tailored to representing individual COCOA levels.
These automata have suffix-language tracking embedded in their definition,
are $L$-congruent-prefix invariant (just like COCOA-level languages (Fact~\ref{fact:simL-implies-simLl})),
and have the simpler safety acceptance condition, although with a ``run jump''.
We then annotate the states of $\Gl$ with suffix languages as well,
using graph-labelling functions.
Finally, we describe a construction that connects all these objects to build the checker,
prove its correctness,
and describe the overall translation procedure.

\subsection{Suffix-language-tracking machines (aka right-congruence automata)}
\label{sec:SLTM}

Consider an omega-regular language $L$ over an alphabet $\Sigma$.
Recall that $L^\textit{suff}$ denotes the set of suffix languages of $L$
and $L^\textit{suff}(p)$ denotes the suffix language of $L$ after reading \text{a finite prefix $p$.}

\emph{The suffix-language-tracking machine (SLTM)} for $L$,
also called \emph{right-congruence automaton~\cite{angluin2024constructing}},
is the machine
$\mc M = (\Sigma, S, s_0, \delta^S)$
with the set of states $S = L^\textit{suff}$,
the initial state $s_0 = L$,
the transition function $\delta^S : S \x \Sigma \to S$
satisfying for every state $s$ and prefix $p$:
$\delta^S(s_0,p) = L^\textit{suff}(p)$,
where we overload the one-letter-successor function $\delta^S$ to accept finite prefixes.
Note that for every pair of prefixes $p_1,p_2$:
$p_1 \sim_L p_2 ~\Iff~ \delta^S(s_0,p_1) = \delta^S(s_0,p_2)$.
For a language represented as a weak alternating automaton,
the SLTM has at most a doubly exponential number of states,
due to existence of a deterministic (parity, Rabin, etc.) automaton recognizing the language
with a doubly exponential number of states.
Figure~\ref{fig:SLTM} gives an example of an SLTM.

\begin{figure}
  \centering
  \begin{minipage}[t]{0.752\linewidth}
    \caption{The SLTM for the language of
             $\varphi = \LTLG(a \impl \LTLX \neg a) \land (\LTLGF a \impl \LTLGF b)$.
             State $s_0$ encodes $L(\varphi)$,
             $s_1$ -- $L(\neg a \land \varphi)$, and
             $s_2$ -- the empty language.
             It is also the SLTM for $L(\neg\varphi)$:
             $s_0$ encodes $L(\neg\varphi)$,
             $s_1$ -- $L(a \lor \neg \varphi)$, and
             $s_2$ -- $\Sigma^\omega$.%
             }
    \label{fig:SLTM}
  \end{minipage}\hspace{0mm}
  \begin{minipage}[t]{0.24\linewidth}
    \vspace{-0.5mm}
    \begin{tikzpicture}[->,>=stealth', auto, node distance=15mm]
      \tikzstyle{every state}=[text=black,font=\footnotesize,initial text={},minimum size=5mm, inner sep=0mm]
      \tikzset{
        every edge/.append style = {font=\footnotesize}
      }

      \node[state,initial left] (s0) {$s_0$};

      \node[state] (s1) [right=6mm of s0] {$s_1$};
      \node[state] (s2) [right=4mm of s1] {$s_2$};

      \path[->] (s0) edge [loop,out=75,in=105,looseness=8] node[above] {$\neg a$} (s0);
      \path[->] (s0) edge [bend right=10] node[below]{$a$} (s1);
      \path[->] (s1) edge [bend right=10] node[above]  {$\neg a$} (s0);
      \path[->] (s1) edge [] node  {$a$} (s2);
      \path[->] (s2) edge [loop,out=75,in=105,looseness=8] node[above] {$*$} (s2);
    \end{tikzpicture}
  \end{minipage}
\end{figure}

The paper~\cite{angluin2024constructing} describes a simple polynomial-time algorithm to compute a machine that is isomorphic to the SLTM for $L(\mc D)$
where $\mc D$ is a deterministic automaton of one of the standard types (B\"uchi, parity, etc.).
The algorithm relies on the ability to check the language equivalence of two deterministic automata,
which can be done in polynomial time.
This algorithm is straightforward to adapt to allow for weak alternating automata as input,
where the equivalence check of two weak alternating automata can be implemented
using complementation, conjunction, and nonemptiness check.
The time complexity is doubly exponential,
due to the doubly exponential worst-case number of states in SLTMs.

\begin{fact}[follows from \cite{angluin2024constructing}]\label{fact:SLTM}
  Given a weak alternating automaton $\mc A$ with $n$ states,
  one can compute in time $\twoexp$\!\! a machine with $\twoexp$\!\! states
  that is isomorphic to the SLTM for $L(\mc A)$.
\end{fact}

\begin{proof}[Proof idea]
  For completeness, we describe the algorithm adapted from \cite[Section~15]{angluin2024constructing}.
  Fix a weak alternating automaton $\mc A = (\Sigma, Q, q_0, \delta, F)$,
  where $\delta: Q\x\Sigma \to \mc B^+(Q)$.
  We compute the following machine $\mc M = (\Sigma, Q^S\!, q_0^S, \delta^S)$.
  Each state is encoded as a formula from $\mc B^+(Q)$.
  The initial state is $q_0^S = q_0$.
  The set of states $Q^S$ and the transition structure $\delta^S$ are grown incrementally
  starting from $Q^S=\{q_0\}$ and the empty $\delta^S$.
  As long as the current set $Q^S$ contains a state $\varphi$ that for some letter $x$ has no successor in $\delta^S$, do the following.
  \li
  \- Define the candidate successor state $\psi_c$ by substituting in $\varphi$ every proposition $q$ with the formula $\delta(q,x)$.

  \- Check if the candidate state $\psi_c$ is language equivalent to some state $\psi$ already present in $Q^S$.
     This can be done by checking language nonemptiness of the alternating automaton starting from the state
     $(\psi_c \land \neg\psi) \lor (\neg\psi_c \land \psi)$.
     This automaton can be computed using simple and-or operations on alternating automata,
     because the automata for $\psi_c$ and $\psi$ are weak and hence their complementation is done by dualizing.

  \- If such $\psi$ exists in $Q^S$,
     set $\delta^S(\varphi,x) = \psi$.
     Otherwise, add $\psi_c$ to $Q^S$ and set $\delta^S(\varphi,x) = \psi_c$.
  \il
  The correctness of the computed machine $\mc M$ follows from two claims:
  (1) all suffix languages are considered, and
  (2) the machine has no two different states representing the same suffix language.
  The second claim is ensured due to the check before adding a new state.
  The first claim follows from the property of alternating automata:
  for every state $q \in Q$ and letter $x$,
  $\SuffLang(L(\mc A_q), x) = \delta(q,x)[\text{substitution}]$,
  where the substitution operation replaces each
  proposition $q'$ with $L(\mc A_{q'})$,
  $\wedge$ with $\cap$, and
  $\vee$ with $\cup$,
  and $\mc A_q$ denotes the automaton starting from $q$.
  The time complexity is dominated by the number of states in the SLTM.
\end{proof}

Finally, let us observe that the SLTMs for a language $L$ and its complement $\bar L$
are isomorphic and their corresponding states are complementation of each other.

\subsection{Floating automata}
\label{subsec:floating-automata}

Floating automata~\cite{ehlers2025rerailingautomata}
are an automaton model that captures the properties of ``nice'' HD-tNCAs of~\cite{DBLP:journals/lmcs/RadiK22}
and which is generalized and tailored for reasoning about COCOA.
Recall from Fact~\ref{fact:simL-implies-simLl} that each COCOA-level language is prefix-invariant wrt.\ $L$,
meaning that any prefix of a word can be replaced with an $L$-congruent prefix.
Moreover, the level languages are co-B\"uchi, and
every co-B\"uchi language is accepted by an automaton that is ``eventually safety''.
Floating automata have both these properties:
$L$-congruent-prefix invariance and the safety acceptance condition.
They also track the current suffix language wrt.\ $L$, aiding the correctness proofs.
Finally,
computing an HD-tNCA from a floating automaton can be done in polynomial time.
We now define the model.

%

Given an SLTM $\mc M = (\Sigma,S,s_0,\delta^S)$ for some omega-regular language,
a deterministic \emph{floating automaton} (\DFA) is a tuple
$\mc F = (\Sigma, Q, \delta, f)$
that consists of
an alphabet $\Sigma$,
a set of states $Q$,
a partial transition function $\delta : Q \x \Sigma \pto Q$, and
a labelling function $f : Q \to S$ that satisfies
$f(\delta(q,x)) = \delta^S(f(q),x)$,
for every $q \in Q$ and $x \in \Sigma$.
Thus,
every state of $\mc F$ is labelled by a single state of the SLTM, and
the labelling is consistent with the SLTM's behaviour.
Note that \DFAs have no initial state;
their runs will later be defined with the aid of the SLTM.

A \emph{path} of $\mc F$
on a finite or infinite (suffix) word $w = x_1 x_2 \ldots$
is a finite or infinite sequence of transitions
$\pi = (q_1,x_1,q_2) (q_2,x_2,q_3) \ldots$
satisfying
$q_{i+1}=\delta(q_i,x_i)$ for every $i \in \{1,\ldots,|w|\}$ for finite paths and for all $i \geq 1$ otherwise.
Sometimes we omit the letters $x_i$ and write the path as $q_1 q_2 \ldots$.
A word may induce no path, a single path, or multiple paths depending on the choice of the initial path state.

Given a finite or infinite word $w = x_1 x_2 \ldots$,
suppose there is a moment $m\in\{1,\ldots,|w|\}$ (for finite $w$) or $m \geq 1$ (when $w$ is infinite)
and a path
$(q_m,x_m,q_{m+1}) (q_{m+1},x_{m+1},q_{m+2}) \ldots$ on the suffix $x_m x_{m+1} \ldots$
starting in state $q_m$ such that $f(q_m) = \delta^S(s_0,x_1\ldots x_{m-1})$;
this path is called a \emph{run} of the \DFA for the original word.
The \emph{length} of a run is the length of the path;
it is either the infinity or $|w|-m+1$.
A word can have several runs, due to the different possible choices of the moment $m$ and the initial path state.

A \DFA $\mc F$ \emph{accepts} an infinite word $x_1 x_2 \ldots$
if
it has an infinite run for it.
Intuitively,
a word is accepted by $\mc F$
if there is a finite run in the SLTM that ends at some moment $m$ in some state $s$,
then the run jumps into a state of $\mc F$ labelled with $s$,
and evolves in $\mc F$ forever.
The \emph{language} $L(\mc F)$ is the set of accepted infinite words.

The definition of nondeterministic floating automata (\NFA) is derived from that of $\text{\DFAs}$ in the standard way,
using a transition relation $\delta \subseteq Q \x \Sigma \x Q$ instead of a transition function.


Observe that we can always remove transient states and transitions (those that do not belong to any SCC) from a floating automaton
without changing its language.
Thus, a floating automaton can be viewed as a collection of disjoint maximal SCCs.

Figure~\ref{fig:FlN} shows example \DFAs wrt.\ the SLTM in Figure~\ref{fig:SLTM}
(the \DFAs in the figure are also \NFAs constructed by Definition~\ref{def:FlN},
 but that is irrelevant here).

\begin{lemma} \label{lem:DFA-to-HD-tNCA}
  Every \DFA can be translated into an HD-tNCA recognizing the same language in polynomial time
  (the input includes the SLTM).
\end{lemma}

\begin{proof}
  We translate the \DFA $\mc F$ to a HD-tNCA $\mc C$ recognizing the same language as follows.
  The states of $\mc C$ are the union
  of the states $S$ of the SLTM $(\Sigma,S,s_0,\delta^S)$ and
  of the states $Q^F$ of the \DFA $(\Sigma, Q^F\!,\delta^F\!,f)$.
  The initial state is the initial state of the SLTM.
  The accepting transitions are all transitions of the \DFA.
  We add rejecting transitions of four kinds:
  transitions between SLTM states;
  transitions from SLTM states to \DFA states;
  transitions between \DFA states; and
  transitions from \DFA states to SLTM states.
  The transitions of the SLTM are added to $\mc C$ as rejecting transitions.
  We add a rejecting transition $s \trans{x} q^F$
  for every $s \in S$, $q^F \in Q^F\!$, $x \in \Sigma$
  such that
  $s \trans{x} f(q^F)$ is a transition of the SLTM.
  We add a rejecting transition $q^F \trans{x} q'^F$
  for every $q^F\!\!,q'^F\in Q^F$ and letter $x$
  such that
  $f(q^F) \trans{x} f(q'^F)$ is a transition of the SLTM.
  Finally,
  add a rejecting transition
  $q^F \trans{x} \delta^S(f(q^F),x)$,
  for every $q^F$ and letter $x$.
  It is straightforward to see that $L(\mc F) = L(\mc C)$.

  We now describe a strategy to resolve nondeterminism in $\mc C$.
  First, let $\leq$ be a total order on run prefixes of $\mc C$
  that orders the run prefixes according to their length,
  and whenever two prefixes have the same length,
  they are ordered in an arbitrarily but fixed way.
  The strategy, on reading a letter, always picks an accepting transition when possible.
  Observe that $\mc C$ never has a nondeterministic choice between two accepting transitions -- only between rejecting transitions and perhaps a single accepting transition.
  When a nondeterministic choice occurs,
  the strategy picks a successor state $q'$ for which there exists a run prefix $\pi_{rej} \cdot \pi_{acc}$ ending in $q'$ on the current word prefix,
  where $\pi_{rej}$ ends in a rejecting transition and $\pi_{acc}$ continues where $\pi_{rej}$ left and consists solely of accepting transitions,
  and where $\pi_{rej}$ is the smallest wrt.\ the order $\leq$ among such prefixes for all possible successors.
  In other words, the strategy picks a successor for which there exists a run prefix on the current word prefix ending with the longest uninterrupted strike of accepting transitions;
  the draw between equally long prefixes is resolved according to the order fixed beforehand.

  We now prove that, given a word accepted by $\mc C$,
  the strategy generates an accepting run.
  The proof resembles that of \cite[Thm.6]{DBLP:conf/fsttcs/EhlersS22}.
  It consists of two parts.
  First,
  we show that the strategy enumerates run prefixes in strictly increasing order:
  as the run progresses,
  the ordering position of the prefix $\pi_{rej}$ that decides the choice of the successor strictly increases.
  Second,
  we prove that eventually the strategy decides the successor using a prefix of an accepting run.

  We prove the first item.
  Consider two arbitrary moments $m<n$ at which the strategy has to resolve nondeterminism.
  At moment $m$,
  the strategy picks the successor $q_{m+1}$
  for which there exists run prefix $\pi^m_{rej}\cdot \pi^m_{acc}$ of length $m$ on reading the word prefix that ends in $q_{m+1}$,
  where $\pi^m_{acc}$ consists solely of accepting transitions, and $\pi^m_{rej}$ is the smallest among such prefixes for such successor states.
  Similarly,
  we define $\pi^n_{rej}\cdot \pi^n_{acc}$ for the moment $n$.
  We now show by contradiction that $\pi^m_{rej} < \pi^n_{rej}$.
  Having $\pi^m_{rej} = \pi^n_{rej}$ implies that $\pi^m_{acc} \sqsubset \pi^n_{acc}$ (a strict prefix),
  hence the transition at the moment $m$ in the current run must be accepting, but it is rejecting.
  Having the dual case $\pi^m_{rej} > \pi^n_{rej}$ contradicts the strategy choice at the moment $n$.

  We prove the second item.
  Fix a word accepted by $\mc C$.
  For this word,
  there exists an accepting run of the form $\rho^\star = \pi^\star_{rej}\cdot\rho^\star_{acc}$,
  where the infinite run suffix $\rho^\star_{acc}$ consists solely of accepting transitions,
  and the run prefix $\pi^\star_{rej}$ is the smallest among such prefixes.
  We show that eventually the strategy's run $\rho$ simulates $\rho^\star$.
  Consider the moment $|\pi^\star_{rej}|$,
  and suppose the strategy's run does not follow $\rho^\star_{acc}$ at that moment (otherwise we are done).
  There exists a moment $m\geq |\pi^\star_{rej}|$ (pick the earliest),
  when the strategy has to resolve nondeterminism (otherwise, $\pi^\star_{rej}$ would not be the smallest).
  When this happens,
  $\pi^\star_{rej}$ is among the prefixes affecting the choice of the successor.
  If the strategy picks $\pi^\star_{rej}$, we are done.
  If it does not, then it picks $\pi_{rej} < \pi^\star_{rej}$.
  When it visits a rejecting transition later (otherwise, $\pi^\star_{rej}$ would not be the smallest),
  it has to pick again.
  At that later moment, $\pi^\star_{rej}$ is still among the choices while $\pi_{rej}$ is not.
  As previously proven, the new choice has a higher position in the order than $\pi_{rej}$.
  And so on.
  Since there is only finitely many elements between $\pi_{rej}$ and $\pi^\star_{rej}$,
  eventually the strategy picks $\pi^\star_{rej}$ and rerails its run to $\rho^\star$.
\end{proof}

We conclude this section with an observation about the general shape of languages recognized by \NFAs.
Recall from Fact~\ref{fact:simL-implies-simLl} that the languages of individual levels of COCOA $(\mc A^1,\ldots,\mc A^k)$ for $L$ are insensitive to $L$-congruent-prefix replacement:
for every $\ell \in \{1,\ldots,k\}$,
prefixes $p_1,p_2$ s.t.\ $p_1 \sim_L p_2$, and suffix $w$:
$p_1 w \in L(\mc A^\ell) \Impl p_2 w \in L(\mc A^\ell)$.
\NFAs possess a similar property wrt.\ the language of the SLTM.
Fix an omega-regular language $L$.
Given an \NFA with respect to an SLTM for $L$,
any two words that start with $L$-congruent prefixes and share the same suffix,
i.e., $p_1 w$ and $p_2 w$ where $p_1 \sim_L p_2$,
are either both accepted by the \NFA or both rejected.
The reason is that the SLTM reaches the same state after reading $p_1$ and $p_2$.
Hence, if $p_1 w$ has an accepting run in the \NFA,
then we can simulate this run on $p_2 w$.
Thus,
replacing a prefix of a word with any prefix congruent wrt.\ $L$ does not affect whether the word is accepted by the \NFA.

\subsection{Obligation graphs and labelling $\lambda^\mc G$ and $\lambda^\tmcG$ of SLTM states}
\label{sec:obligation-graphs}
\label{sec:graph-labelling}

An \emph{obligation graph} for an omega-regular language $L$ is any tNBA recognizing it.
Given a weak alternating automaton $\mc A$ with $n$ states,
the Miyano-Hayashi construction~\cite{MHConstruction} removes universal transitions and yields an obligation graph for $L(\mc A)$.
Since weak alternating automata can be complemented via dualization,
in the same way we can get an obligation graph for $\overline{L(\mc A)}$.
In both cases, the resulting obligation graph has $2^{O(n)}$ states.

Given an SLTM $(\Sigma,Q^S,q_0^S,\delta^S)$ for an omega-regular language $L$ and an obligation graph $(\Sigma, V, v_0, \delta^\mcG,\alpha)$ for $L$,
the \emph{graph labelling} of the SLTM is a function
$\lambda^\mcG: Q^S \to 2^V$
that maps every $s \in Q^S$ to the set
$\{v \mid \exists p\in \Sigma^*: s = \delta^S(s_0,p) \land v \in \delta^\mcG(v_0,p)\}$.
In other words,
every SLTM state is mapped to the set of vertices reachable on the class of congruent prefixes that the SLTM state encodes.
The graph labelling $\lambda^\mcG$ can be constructed in polynomial time in the size of given SLTM and obligation graph,
by first building the product between the SLTM and the graph, then checking which states are paired together.

\begin{lemma}\label{lem:graph-labelling}
  Given an SLTM and an obligation graph $\mcG$,
  the graph labelling $\lambda^\mcG$ can be constructed in polynomial time.
\end{lemma}

Figures~\ref{fig:NBA}~and~\ref{fig:NBA_neg} show examples of obligation graphs and graph labelling.

\begin{figure}
  \centering
  \begin{minipage}[t]{0.4\linewidth}
    \vspace{0mm}
    \begin{tikzpicture}[->,>=stealth', auto, node distance=15mm]
      \tikzstyle{every state}=[text=black,font=\footnotesize,initial text={},minimum size=5mm, inner sep=0mm]
      \tikzset{
        every edge/.append style = {font=\footnotesize}
      }
      \tikzstyle{marked}=[
        draw=green!38!black,
        postaction={
          decorate,
          decoration={
            markings,
            mark=at position 0.5 with {
              \fill[green!38!black] (0,0) circle[radius=0.4mm];
            }
          }
        }
      ]

      \node[state,initial above] (v0) {$v_0$};
      \node[state] (v1) [right=14mm of v0] {$v_1$};
      \node[state] (v2) [left=8mm of v0] {$v_2$};

      \path[->] (v0) edge [marked,loop,out=240,in=210,looseness=8] node[xshift=0mm,yshift=1mm] {$\bar a b$} (v0)
                (v0) edge [loop,out=285,in=255,looseness=8] node[below] {$\bar a \bar b$} (v0)
                (v0) edge node[above]{$\neg a$} (v2)
                (v2) edge [marked,loop,out=195,in=165,looseness=8] node[left] {$\neg a$} (v2)
                (v0) edge [marked,bend right=6] node[below,yshift=0.8mm]{$a b$} (v1)
                (v0) edge [bend right=40] node[below,yshift=0.9mm]{$a \bar b$} (v1)
                (v1) edge [marked,bend right=6] node[above,yshift=-0.8mm]  {$\bar a  b$} (v0)
                (v1) edge [bend right=40] node[above,yshift=-0.8mm]  {$\bar a  \bar b$} (v0);

      \tikzstyle{rect}=[rectangle, rounded corners, text=black,font=\footnotesize,initial text={},minimum size=3.7mm, inner sep=0.6mm]
    \end{tikzpicture}
  \end{minipage}
  \begin{minipage}[t]{0.3\linewidth}
      \vspace{0mm}
      \footnotesize
      \begin{align*}
        \hspace*{-\mathindent}
        &\textit{Graph labelling of the SLTM:}\\[-0.5mm]
        &\lambda^\mc G(s_0) = \{v_0,v_2\} \\[-0.5mm]
        &\lambda^\mc G(s_1) = \{v_1\} \\[-0.5mm]
        &\lambda^\mc G(s_2) = \{\}
      \end{align*}
  \end{minipage}
  \caption{An obligation graph $\mcG$ (tNBA) recognizing the language of
           $\varphi = \LTLG(a \impl \LTLX \neg a) \land (\LTLGF a \impl \LTLGF b)$,
           and graph labelling $\lambda^\mcG$ wrt.\ the SLTM for $L(\varphi)$ in Figure~\ref{fig:SLTM}.
           Notation such as $a\bar b$ denotes $a\land \neg b$.
           Accepting transitions are marked by a dot.%
           }
  \label{fig:NBA}
\end{figure}
\begin{figure}
  \centering
  \begin{minipage}[t]{0.3\linewidth}
    \vspace{0mm}
    \begin{tikzpicture}[->,>=stealth', auto, node distance=15mm]
      \tikzstyle{every state}=[text=black,font=\footnotesize,initial text={},minimum size=5mm, inner sep=0mm]
      \tikzset{
        every edge/.append style = {font=\footnotesize}
      }
      \tikzstyle{marked}=[
        draw=green!38!black,
        postaction={
          decorate,
          decoration={
            markings,
            mark=at position 0.5 with {
              \fill[green!38!black] (0,0) circle[radius=0.4mm];
            }
          }
        }
      ]

      \tikzstyle{rect}=[rectangle, rounded corners, text=black,font=\footnotesize,initial text={},minimum size=3.7mm, inner sep=0.6mm]

      \node[state,initial left] (nv0) {$\tilde v_0$};
      \node[state] (nv1) [above right =4mm of nv0] {$\tilde v_1$};
      \node[state] (nv2) [right=7mm of nv1] {$\tilde v_2$};
      \node[state] (nv3) [below right =4mm of nv0] {$\tilde v_3$};

      \path[->]
           (nv0) edge [loop,out=75,in=105,looseness=8] node[above,xshift=0mm,yshift=0mm] {$*$} (nv0)
           (nv0) edge node[below,xshift=1mm,yshift=0.7mm] {$a$} (nv1)
           (nv1) edge node[below,xshift=0mm,yshift=0.5mm] {$a$} (nv2)
           (nv2) edge [marked,loop,out=-15,in=15,looseness=8] node[right,xshift=0mm,yshift=0mm] {$*$} (nv2)
           (nv0) edge node[above,xshift=1mm,yshift=-0.7mm] {$*$} (nv3)
           (nv3) edge [loop,out=45,in=75,looseness=8] node[right,xshift=0mm,yshift=0mm] {$\bar a \bar b$} (nv3)
           (nv3) edge [marked,loop,out=-15,in=15,looseness=8] node[right,xshift=0mm,yshift=0mm] {$a \bar b$} (nv3)
      ;



    \end{tikzpicture}
  \end{minipage}
  \begin{minipage}[t]{0.3\linewidth}
      \vspace{-2mm}
      \footnotesize
      \begin{align*}
        \hspace*{-\mathindent}
        &\textit{Graph labelling of the SLTM:}\\[-0.5mm]
        &\lambda^{\tilde {\mc G}}(s_0) = \{\tilde v_0,\tilde v_3\} \\[-0.99mm]
        &\lambda^{\tilde {\mc G}}(s_1) = \{\tilde v_0,\tilde v_1,\tilde v_3\} \\[-0.99mm]
        &\lambda^{\tilde {\mc G}}(s_2) = \{\tilde v_0,\tilde v_2,\tilde v_3\}
      \end{align*}
  \end{minipage}
  \caption{An obligation graph $\tmcG$ (tNBA) for the negation $\neg \varphi$ of the formula used in Figure~\ref{fig:NBA}, and the graph labelling $\lambda^\tmcG$ of the SLTM for $L(\neg\varphi)$ in Figure~\ref{fig:SLTM}.
           Notation such as $a\bar b$ denotes $a\land \neg b$.
           Accepting transitions are marked by a dot.%
           }
  \label{fig:NBA_neg}
\end{figure}

\subsection{Incremental construction of a \DFA for the COCOA level-$\ell$ language}
\label{subsec:FCW_Al}

This section describes a procedure to construct an DfA for the $\ell$th level of the COCOA,
given an DfA for the previous level.
We first construct a nondeterministic floating automaton $\mc F^\ell_N$,
prove that it indeed recognizes the level-$\ell$ of the COCOA,
then determinize it into $\mc F^\ell$.

Fix an omega-regular language $L$,
and let $\mc M = (\Sigma,S,s_0,\delta^S)$ be an SLTM for $L$.

Let $\mc F^0$ be an DfA wrt.\ $\mc M$ recognizing $\Sigma^\omega$;
it is not part of the COCOA.

Fix $\ell \geq 1$.
We use the following notation.
When $\ell$ is even,
$\Gl$ and $\lambda^\Gl$ denote an obligation graph (NBA) and the vertex labelling for $L$ wrt.\ $\mc M$;
when $\ell$ is odd -- for $\bar L$.
The NfA $\mc F^{\ell}_N$ is defined inductively as a product of DfA $\mc F^{\ell-1}$ and NBA $\Gl$
with a specially defined acceptance condition.

\begin{definition}[NfA $\mc F^\ell_N$]
  \label{def:FlN}
  Given an SLTM $\mc M = (\Sigma,S,s_0,\delta^S)$,
  the obligation graph $\Gl = (\Sigma, V,v_0,\delta^\Gl\!,\alpha^\Gl)$ and the graph labelling $\lambda^\Gl$, and
  a DfA $\mc F^{\ell-1}\! = (\Sigma,Q^{\ell-1}\!,\delta^{\ell-1}\!, f^{\ell-1})$ wrt.\ $\mc M$,
  the \NFA $\mc F^\ell_N = (\Sigma,Q^\ell\!,\delta^\ell\!,f^\ell)$ is constructed in two steps.

  \smallskip
  \noindent
  Step 1. Define the transition structure as a product of DfA $\Fellmo$ and NBA $\Gl$:
  \li
  \- $Q^\ell = \{ (q^{\ell-1}\!,v) \mid q^\ellmo \in Q^\ellmo\!, v \in \lambda^\Gl(f^{\ell-1}(q^{\ell-1})) \}$,\\
     i.e.,
     $q^\ellmo$ is paired with every vertex $v$ that can be reached on the same prefix as $q^\ellmo$.

  \- $f^\ell((q^{\ell-1}\!,v)) = f^{\ell-1}(q^{\ell-1})$ for every $(q^{\ell-1}\!,v) \in Q^\ell$.

  \- For every $(q^{\ell-1}\!, v) \in Q^\ell$ and letter $x$:~
      $\delta^\ell((q^{\ell-1}\!,v),x) ~=~
      \{\delta^{\ell-1}(q^{\ell-1}\!,x)\} \x \delta^\Gl(v,x)$.
  \il

  \vspace{-0.8mm}
  \noindent
  Step 2. Define acceptance:
  \li
  \- Decompose the transition structure into maximal SCCs
     $\{(Q_1^\ell,\delta_1^\ell), \ldots, (Q_n^\ell,\delta_n^\ell) \}$.

  \- The states and transitions outside of the SCCs are removed from $Q^\ell$ and $\delta^\ell$.

  \- Consider an SCC $(Q^\ell_i,\delta^\ell_i)$.
     We keep the SCC $(Q^\ell_i,\delta^\ell_i)$ in $\mc F^{\ell}_N$
     if $\delta^\ell_i$ contains a transition $(q^{\ell-1}\!,v) \trans{x} ({q^\ellmo}'\!,v')$
     s.t.\ $v\trans{x} v'$ is a $\Gl$-accepting transition,
     otherwise we remove the SCC from $\mc F^{\ell}_N$.
  \il
\end{definition}

\noindent We give an intuition.
\li
\- The purpose of the product is to check the condition
   $\forall J.\exists \text{SLII}^L_{(J,w)} w' \in L(\mc F^\ellmo) \cap L(\Gl)$.

\- By using $\Fellmo$ in the product,
   we consider only the words of $\Fellmo$.

\- The product of \DFA $\Fellmo$ with NBA $\Gl$ pairs together the states and vertices that can be reached on the same prefix, which is done with the aid of the graph labelling function $\lambda^\Gl$.
   Recall that $\lambda^\Gl$ maps an SLTM state to the set of vertices of $\Gl$
   reachable on some prefix of the prefix-congruence class that the SLTM state represents.
   The resulting floating automaton of Step 1 -- if we require its run to also satisfy the B\"uchi acceptance condition of the NBA component -- checks whether a given word $w$ belongs to $L(\Fellmo)\cap L(\Gl)$.
   We, however, need to check $\forall J.\exists \text{SLII}^L_{(J,w)} w' \in L(\mc F^\ellmo) \cap L(\Gl)$,
   which is achieved in Step 2.

\- In Step 2, we keep only those SCCs that contain a transition whose NBA projection is accepting.
   If an SCC contains such a transition,
   then on reaching the SCC,
   when a moment from $J$ arrives,
   by inserting a finite word,
   we can visit that NBA accepting transition and come back.
   Since the states of the product track the SLTM state,
   returning to where we started means that the inserted finite word preserves the suffix language.
   We inject such words for every moment in $J$.
   Thus, if a word $w$ ends in such an SCC,
   then it satisfies $\forall J.\exists \text{SLII}^L_{(J,w)} w' \in L(\mc F^\ellmo) \cap L(\Gl)$
   (Lemma~\ref{lem:Ll_includes_AlN}).
   Consider the other direction,
   and suppose a word $w$ satisfies the for-all-$J$ condition.
   We pick a set $J$ of points separated sufficiently far apart.
   There exists an SLII$^L_{(J,w)}$ $w'$ having an accepting run in both $\Fellmo$ and $\Gl$,
   hence $w'$ has an accepting run in $\mc F^\ell_N$,
   thus $w'$ ends in an SCC preserved by Step 2.
   Since the injection points are far apart,
   the run of $\mc F^\ell_N$ visits a certain state twice
   while reading a part of $w'$ between the injections, i.e., a part of the original $w$.
   By looping the run between those visits,
   we create an infinite and hence accepting run of $\Fl_N$ on the original $w$.
   This direction constitutes Lemma~\ref{lem:AlN_includes_Ll}.
\il

\begin{lemma}\label{lem:Ll_includes_AlN}
  For each level index $\ell$ of the COCOA for $L$:
  \begin{equation*}
    L(\mc F^\ell_N)
    ~\subseteq~
    \big\{w \in L(\mc F^{\ellmo}) \mid \forall J.\,\exists \textnormal{SLII}^{L}_{(J,w)} w': w' \in L(\mc F^{\ell-1}) \cap L(\Gl)\big\}.
  \end{equation*}
\end{lemma}

\begin{proof}
  Consider a word $w = x_0 x_1 \ldots$ accepted by $\mc F^\ell_N$.
  In what follows, for a given index set $J$,
  we construct an SLII $w'$ accepted by both $\Gl$ and $\mc F^{\ell-1}$.

  Since $w \in L(\mc F^\ell_N)$,
  there exist
  a moment $m$ and
  an infinite path $\rho^\ell_{w,m} = q^\ell_m q^\ell_{m+1}\ldots$ of $\mc F^\ell_N$ on $x_m x_{m+1} \ldots$
  such that
  $f^\ell(q^\ell_m) = \delta^S(s_0,x_0 \ldots x_{m-1})$.
  Let $q^\ell_J$ be an arbitrary state that appears infinitely often in this witness path.
  Let $J' \subseteq J$ be the subset of moments of $J$ that are greater than $m$.
  By the definition of $\mc F^\ell_N$,
  the SCC containing $q^\ell_J$ also contains a transition
  $\tau^\ell = ( q^{\ell-1}\!, v) \trans{y} ({\hat q^{\ell-1}}\!,\hat v)$,
  where $ v \trans{y} \hat v$ is a $\Gl$-accepting transition.
  There exists a finite path $\pi$ on some finite word $u$
  that starts in the state $q^\ell_J$,
  goes through the transition $\tau^\ell$,
  and returns to $q^\ell_J$.
  Notice that the SLTM label is the same at the beginning and the end of this loop.
  Therefore, extending any prefix word
  whose path ends in $q^\ell_J$
  by $u$ preserves the suffix languages wrt.\ $L$ and $\overline{L}$.
  Let $w'$ be the result of inserting $u$ into $w$ at each moment in $J'$ and nothing at moments in $J\setminus J'$;
  then $w' \in \text{SLII}^{L}_{(J,w)}$.
  Let $\rho^\ell_{w'\!,m}$ be the infinite path of $\mc F^\ell_N$ obtained by inserting $\pi$ into $\rho^\ell_{w,m}$ at each moment in $J'$:
  the infinite path $\rho^\ell_{w'\!,m}$ and the moment $m$ witness the acceptance of $w'$ by $\mc F^\ell_N$.
  Hence,
  by the definition of $\mc F^\ell_N$,
  the $F^{\ell-1}$-projection of $\rho^\ell_{w'\!,m}$ and the moment $m$ witness the acceptance of $w'$ by $\mc F^{\ell-1}$
  (note that $\delta^S(s_0,x_0,\ldots,x_{m-1}) = f^\ell(q^\ell_m) = f^{\ell-1}(q^{\ell-1}_m)$,
   where $q^\ell_m = (q^{\ell-1}_m\!,v_m)$).
  Finally, we show that $w' \in L(\Gl)$.
  Since $v_m \in \lambda^\Gl(f^\ell(q^\ell_m))$,
  there is a prefix $p$
  such that
  $p \sim_{L} x_0\ldots x_{m-1}$ and $v_m$ is reached from $v_0$ on reading $p$.
  The suffix $s$ of $w'$ starting at moment $m$ has an infinite path in $\Gl$ that visits an accepting transition (namely, $v \trans{y} \hat v$) infinitely often,
  hence $p \cdot s \in L(\Gl)$, thus $\Gl$ also accepts $x_0\ldots x_{m-1} \cdot s = w'$.
\end{proof}

\begin{lemma}\label{lem:AlN_includes_Ll}
  Fix a level index $\ell$ of the COCOA for $L$.
  Under the assumption that \DFA $\mc F^{\ellmo}$ represents the language of the level $(\ellmo)$ of the COCOA, or $\Sigma^\omega$ for $\ell=1$, we have:
  \begin{equation*}
    \big\{w \in L(\mc F^{\ell-1})\mid \forall J.\,\exists \textnormal{SLII}^{L}_{(J,w)} w': w' \in L(\mc F^{\ell-1}) \cap L(\Gl)\big\}
    ~\subseteq~
    L(\mc F^\ell_N).
  \end{equation*}
\end{lemma}

\begin{proof}
  First, we observe that the left and the right part of the lemma equation are omega-regular languages.
  The language $L(\Dl_N)$ is omega-regular since $\Dl_N$ is an \NFA.
  Since $\mc F^{\ell-1}$ represents the language of the level $(\ell{-}1)$ of the COCOA,
  the left part of the lemma equation represents exactly the language of the level $\ell$ of the COCOA,
  which is omega-regular.

  The proof of the lemma reasons about ultimately-periodic words.
  A word is \emph{ultimately-periodic} if it is of the form $uv^\omega$ for $u,v\in\Sigma^*$.
  Given a language $L$,
  let $L^\mathit{up}$ denote the set of all ultimately-periodic words of $L$.

  We now state three claims, then proceed to the proof.

  \smallskip
  \noindent
  \emph{Claim 1 \cite[Fact 1]{10.1007/3-540-58027-1_27}:}
  For every two omega-regular languages $L_1$ and $L_2$:
  $L_1^\mathit{up} \subseteq L_2^\mathit{up} \Rightarrow L_1 \subseteq L_2$.

  \smallskip
  \noindent
  \emph{Claim 2:}
  Let $w$ be an ultimately-periodic word and let $\mc F$ be an \NFA.
  Suppose that, for every $l \in \bbN$,
  $\mc F$ has a run of length at least $l$ on some prefix of $w$.
  Then $w \in L(\mc F)$.

  \noindent
  Proof:
  Fix a word $w = uv^\omega$.
  Choose $l = |u| + |v|\cdot|Q^\mc F| + 1$,
  where $|Q^\mc F|$ is the number of states in $\mc F$.
  Consider a run $\rho$ of $\mc F$ of length at least $l$ on some prefix of $w$.
  By the pigeonhole principle,
  on reading the first letter of $v$ some state repeats.
  Let $j<k$ be the moments of two occurences of that state.
  We construct an infinite run of $\mc F$ on $w$
  by following $\rho$ up to the moment $j$ and then looping the transitions from $j$ till $k$.

  \smallskip
  \noindent
  \emph{Claim 3:} Every word $w \in L(\Gl)\cap L(\mc F^{\ell-1})$ is accepted by $\mc F^\ell_N$.

  \noindent
  Proof:
  Fix a word $w=x_1 x_2 \ldots \in L(\Gl)\cap L(\mc F^\ellmo)$.
  Let $m$ be a moment and $q^\ellmo_m q^\ellmo_{m+1} \ldots$ be an infinite run of $\mc F^\ellmo$
  witnessing the acceptance of $w$.
  Let $\rho^\Gl = v_1 v_2 \ldots$ be an accepting run of $\Gl$ on $w$.
  It is easy to see that the infinite sequence
  $(q^\ellmo_m\!, v_m) \trans{x_m} (q^\ellmo_{m+1}, v_{m+1}) \trans{x_{m+1}} \ldots$
  forms successive transitions of $\mc F^\ell_N$,
  hence $\mc F^\ell_N$ has an infinite run on $w$,
  so $w \in L(\mc F^\ell_N)$.

  \smallskip
  We are ready to prove the lemma.
  By Claim 1, we focus on ultimately-periodic words.
  \li
  \newcommand\tm{{\tilde m}}
  \- Fix an ultimately-periodic word $w$
     such that
     $\forall J.\exists \textnormal{SLII}^L_{(J,w)} \tilde w: \tilde w \in L(\mc F^{\ell-1}) \cap L(\Gl)$.

  \- Take an arbitrary $l \in \bbN$. Fix some $J$ whose consecutive indices are separated by $l$ steps.

  \- There exists $\tilde w = \tilde x_0 \tilde x_1 \ldots$ as mentioned above.

  \- By Claim 3, $\tilde w \in L(\mc F^\ell_N)$.
     Let $\tilde n$ be a moment and $q^\ell_{\tilde n} \trans{\tilde x_{\tilde n}} q^\ell_{\tilde n+1} \trans{\tilde x_{\tilde n+1}} \ldots$ an infinite run
     witnessing the acceptance of $\tilde w$ by $\mc F^\ell_N$.
     Since the indices in $J$ are separated by $l$ steps,
     there exists a moment $\tm \geq \tilde n$ from which the infinite run reads the original letters (not insertions) for $l$ steps.
     This corresponds to the transitions $q^\ell_{\tm} \trans{\tilde x_\tm}  \ldots \trans{\tilde x_{\tm-1+l}} q^\ell_{\tm+l}$.

  \- Let $\tilde p = \tilde x_0 \ldots \tilde x_{\tm-1}$.
     Let $p = x_0 \ldots x_{m-1} $ be the corresponding prefix of $w$,
     derived from $\tilde p$ by removing the suffix-language-preserving insertions which were inserted into $w$ to get $\tilde w$.
     Note that $\tilde x_{\tm-1+i} = x_{m-1+i}$ for all $i \in \{1,\ldots,l\}$.

  \- Since the previously mentioned run for $\tilde p \cdot \tilde x_\tm \ldots \tilde x_{\tm-1+l}$ starts in $q^\ell_\tm$
     and $\tilde p \sim_{L} p$,
     a run for $p \cdot x_\tm \ldots x_{\tm-1+l}$ can also start from $q^\ell_\tm$.

  \- For the next $l$ steps no insertions are performed,
     so the sequence $q^\ell_\tm \trans{x_m} \ldots \trans{x_{m+l-1}} q^\ell_{\tm+l}$ is actually a run of the length $l$ of $\mc F^\ell_N$ on the prefix $x_0\ldots x_{m-1+l}$ of $w$,
     starting at the moment $m$.
     We constructed a run required by Claim~2, so $w$ is accepted by $F^\ell_N$.
  \il
\end{proof}

\begin{example}\label{example:def:FlN}
  Consider the language of
  $\varphi = \LTLG(a \impl \LTLX \neg a) \land (\LTLGF a \impl \LTLGF b)$
  over the alphabet $\Sigma=2^{\{a,b\}}$.
  Its SLTM is in Figure~\ref{fig:SLTM},
  obligation graphs $\mc G_\varphi$ for $L(\varphi)$ and $\mc G_{\neg\varphi}$ for $L(\neg\varphi)$ are in Figures~\ref{fig:NBA}~and~\ref{fig:NBA_neg}.
  Figure~\ref{fig:FlN} shows the \DFA $\mc F^0$ recognizing $\Sigma^\omega$,
  as well as $\mc F^1_N$ and $\mc F^2_N$ constructed according to Definition~\ref{def:FlN}.
  Let us go through the construction of $\mc F^1_N$ and $\mc F^2_N$.
  The \NFA $\mc F^1_N$ is constructed as the product of the \DFA $\mc F^0$ and the NfBA from Figure~\ref{fig:NBA_neg}
  plus the acceptance-condition analysis.
  Since Step 2 of Definition~\ref{def:FlN} removes transient states and transitions,
  as well as SCCs that do not have a $\mc G$-accepting transition,
  it suffices to consider two SCCs of $\mc G_{\neg\varphi}$:
  the SCC of $\tilde v_2$ and the SCC of $\tilde v_3$.
  Since $\tilde v_2 \in \lambda^{\mc G_{\neg\varphi}}(s_2)$,
  we have $(s_2, \tilde v_2) \in Q^1_N$.
  In the product,
  the state $(s_2, \tilde v_2)$ has a self-loop on every letter.
  Step 2 keeps this SCC in $\mc F^1$,
  because it contains a transition whose projection is accepting in $\mc G_{\neg\varphi}$.
  Now consider vertex $\tilde v_3$.
  Pairing $\tilde v_3$ with all three states of $\mc F^0$
  yields two SCCs in the product,
  and each contains a transition whose projection is accepting in $\mc G_{\neg\varphi}$.
  However,
  we can remove the SCC containing $(s_2,\tilde v_3)$,
  since its language is subsumed by the language of the SCC of $(s_2,\tilde v_2)$.
  Now to the \NFA $\mc F^2_N$.
  The \NFA $\mc F^1_N$ is determintistic,
  so we construct the \NFA $\mc F^2_N$ as the product of $\mc F^1_N$ and $\mc G_{\varphi}$ from Figure~\ref{fig:NBA}.
  The state $(s_2,\tilde v_2)$ of $\mc F^1_N$ cannot be paired with any vertex of $\mc G_{\varphi}$,
  because $f^1_N((s_2,\tilde v_2))=s_2$ and $\lambda^{\mc G_{\varphi}}(s_2) = \emptyset$.
  The state $(s_0,\tilde v_3)$ can be paired with vertices $v_0$ and $v_2$,
  since $f^1_N((s_0,\tilde v_3))=s_0$ and $\lambda^\mc G(s_0) = \{v_0,v_2\}$.
  After Step 2, only the state $(s_0,\tilde v_3,v_2)$ remains,
  yielding $\mc F^2_N$ in the figure.
  Finally, applying Definition~\ref{def:FlN} to $\mc F^2_N$ (which is deterministic) and $\mcG_{\neg\varphi}$ yields the empty-language automaton.
\end{example}

\begin{figure}
  \centering
  \begin{minipage}[t]{0.28\linewidth}
    \vspace{0.7mm}

    \begin{tikzpicture}[->,>=stealth', auto, node distance=15mm,every node/.style={font=\footnotesize}]
      \tikzstyle{every state}=[rectangle, rounded corners, text=black,font=\footnotesize,initial text={},minimum size=4mm, inner sep=0.6mm]
      \tikzset{
        every edge/.append style = {font=\footnotesize}
      }
      \tikzstyle{marked}=[
        draw=green!38!black,
        postaction={
          decorate,
          decoration={
            markings,
            mark=at position 0.5 with {
              \fill[green!38!black] (0,0) circle[radius=0.4mm];
            }
          }
        }
      ]

      \newcommand\StLa[1]{\scalebox{0.8}{(}\hspace{-0.1mm}#1\hspace{-0.2mm}\scalebox{0.8}{)}}
      \node[state] (s0) {$s_0$};
      \node[state] (s1) [right of=s0] {$s_1$};

      \node[state] (s2) [right=0.6cm of s1] {$s_2$};

      \path[->] (s2) edge [loop,out=75,in=105,looseness=10] node[above,xshift=0mm,yshift=-0.5mm] {$*$} (s2)
                (s0) edge [loop,out=75,in=105,looseness=10] node[above,yshift=-0.5mm] {$\neg a$} (s0)
                (s0) edge [bend right=8] node[below,yshift=0.8mm]{$a$} (s1)
                (s1) edge [bend right=8] node[above,yshift=-0.8mm]{$\neg a$} (s0);

      \begin{pgfonlayer}{background}
        \coordinate (topPad) at ([yshift=7.2mm]s0.north);
        \coordinate (bottomPad) at ([yshift=-3.2mm]s0.south);
        \coordinate (leftPad) at ([xshift=-4.5mm]s0.west);
        \coordinate (rightPad) at ([xshift=4.2mm]s2.east);
        \node[fill=gray!15, draw=none, rounded corners, fit=(topPad)(leftPad)(rightPad)(bottomPad), inner sep=0mm] {};
      \end{pgfonlayer}

    \end{tikzpicture}
    \label{fig:FlN:F0}
  \end{minipage}\hspace{3mm}
  \begin{minipage}[t]{0.41\linewidth}
    \vspace{0mm}
    \begin{tikzpicture}[->,>=stealth', auto, node distance=15mm,every node/.style={font=\footnotesize}]
      \tikzstyle{every state}=[rectangle,rounded corners,text=black,font=\footnotesize,initial text={},minimum size=4mm, inner sep=0.6mm]
      \tikzset{
        every edge/.append style = {font=\footnotesize}
      }
      \tikzstyle{marked}=[
        draw=green!38!black,
        postaction={
          decorate,
          decoration={
            markings,
            mark=at position 0.5 with {
              \fill[green!38!black] (0,0) circle[radius=0.4mm];
            }
          }
        }
      ]

      \newcommand\StLa[1]{\scalebox{0.8}{(}\hspace{-0.1mm}#1\hspace{-0.2mm}\scalebox{0.8}{)}}
      \node[state] (s0) {$(s_0,\tilde v_3)$};
      \node[state] (s1) [right=0.9cm of s0] {$(s_1,\tilde v_3)$};

      \node[state] (s2) [right=0.6cm of s1] {$(s_2,\tilde v_2)$};

      \path[->] (s2) edge [loop,out=75,in=105,looseness=10] node[above,xshift=0mm,yshift=-0.5mm] {$*$} (s2)
                (s0) edge [loop,out=75,in=105,looseness=10] node[above,yshift=-0.5mm] {$\bar a\bar b$} (s0)
                (s0) edge [bend right=8] node[below,yshift=0.8mm]{$a \bar b$} (s1)
                (s1) edge [bend right=8] node[above,yshift=-0.8mm]{$\bar a\bar b$} (s0);

      \begin{pgfonlayer}{background}
        \coordinate (topPad) at ([yshift=7.2mm]s0.north);
        \coordinate (bottomPad) at ([yshift=-3.2mm]s0.south);
        \coordinate (leftPad) at ([xshift=-4.5mm]s0.west);
        \coordinate (rightPad) at ([xshift=4.2mm]s2.east);
        \node[fill=gray!15, draw=none, rounded corners, fit=(topPad)(leftPad)(rightPad)(bottomPad), inner sep=0mm] {};
      \end{pgfonlayer}

    \end{tikzpicture}
    \label{fig:FlN:F1}
  \end{minipage}\hspace{3mm}
  \begin{minipage}[t]{0.18\linewidth}
    \vspace{0mm}
    \begin{tikzpicture}[->,>=stealth', auto, node distance=15mm,every node/.style={font=\footnotesize}]
      \tikzstyle{every state}=[rectangle, rounded corners, text=black,font=\footnotesize,initial text={},minimum size=3.7mm, inner sep=0.6mm]
      \tikzset{
        every edge/.append style = {font=\footnotesize}
      }
      \tikzstyle{marked}=[
        draw=green!38!black,
        postaction={
          decorate,
          decoration={
            markings,
            mark=at position 0.5 with {
              \fill[green!38!black] (0,0) circle[radius=0.4mm];
            }
          }
        }
      ]

      \node[state] (s02) {$(s_0,\tilde v_3, v_2)$};

      \path[->]
           (s02) edge [loop,out=75,in=105,looseness=10] node[above,xshift=0mm,yshift=-0.5mm] {$\bar a \bar b$} (s02);

      \begin{pgfonlayer}{background}
        \coordinate (topPad) at ([yshift=7.2mm]s02.north);
        \coordinate (bottomPad) at ([yshift=-3.2mm]s02.south);
        \coordinate (leftPad) at ([xshift=-4.5mm]s02.west);
        \coordinate (rightPad) at ([xshift=4.2mm]s02.east);
        \node[fill=gray!15, draw=none, rounded corners, fit=(topPad)(leftPad)(rightPad)(bottomPad), inner sep=0mm] {};
      \end{pgfonlayer}

    \end{tikzpicture}
  \end{minipage}
  \caption{\DFA $\mc F^0$, and \NFAs $\mc F^1_N$ and $\mc F^2_N$ (they are actually deterministic)
           constructed according to Definition~\ref{def:FlN}.
           All automata are wrt.\ the SLTM from Figure~\ref{fig:SLTM}.
           $\mc F^0$ recognizes the language $\Sigma^\omega$,
           its states are exactly the states of the SLTM,
           and the SLTM labelling $f^0_N$ is the identity function.
           $\mc F^1_N$ recognizes the language of $\neg \LTLG(a \impl \LTLX \neg a) \lor \LTLFG \neg b$.
           $\mc F^2_N$ recognizes the language of
           $\LTLG(a \impl \LTLX \neg a) \land \LTLFG (\neg a \land \neg b)$.
           The labellings $f^1_N$ and $f^2_N$ return the first component of a given state.}
  \label{fig:FlN}
\end{figure}

\subparagraph*{Determinization.}
We now show how to determinize \NFA $\Dl_N$ into a \DFA.
Recall that the nondeterminism of \NFAs has three sources.
First,
deciding when to jump from an SLTM state into a state of the \NFA.
Second,
deciding on the entry state -- into which state of the NFA -- to jump right after leaving the SLTM.
Third,
when already in a state of the \NFA, deciding on a successor when reading a letter.
\DFAs also have nonderministic nature, but in general only of the first and the second kind.
We describe a determinization construction that keeps nondeterminism of the first kind
but removes the second and third reasons.
Note that the LDL-to-COCOA translation of this section naturally handles the first and second kinds of nondeterminism of general-form \DFAs ~--
it is the LDL-to-DPA translation of Section~\ref{sec:LDL-to-DPA} that requires
the entry point to be deterministic.
The idea is to start from a set of entry states instead of a single one,
and apply the standard subset construction.
Let
$\mc F^\ell_N = (\Sigma,Q^\ell_N,\delta^\ell_N,f^\ell_N)$
be the \NFA constructed from $\mc F^\ellmo$ according to Definition~\ref{def:FlN}.
\newcommand\Qlifted{Q^\ellmo_\mathit{lifted}} %
Let
$\Qlifted = \{q^\ellmo \mid \exists v:(q^\ellmo\!,v) \in Q^\ell_N\}$
be the states of $\mc F^\ellmo$ that appear as a component of states of $\mc F^\ell_N$.
For every $q^\ellmo \in \Qlifted$,
define $\Vinit_{q^\ellmo} = \{v \mid (q^\ellmo,v) \in Q^\ell_N\}$.
The set $Q_\mathit{enter}^\ell = \{(q^\ellmo,\Vinit_{q^\ellmo})\mid q^\ellmo \in \Qlifted\}$
is called the set of \emph{entry} states.
Note that each entry state is from $Q^\ellmo \x 2^V$
and can be seen as a set of states of $\mc F^\ell_N$ sharing the same $\mc F^\ellmo$ state.
Since each entry state fixes the $\mc F^\ellmo$ component and $\mc F^\ellmo$ is deterministic,
we can apply the subset construction wrt.\ $\delta^\ell_N$ starting from a given entry state.

\begin{definition}[from \NFA $\mc F^\ell_N$ to \DFA $\mc F^\ell$]\label{def:FlN-to-Fl}
  Let \NFA $\mc F^\ell_N = (\Sigma,Q^\ell_N,\delta^\ell_N,f^\ell_N)$ be constructed according to Definition~\ref{def:FlN}.
  Let $Q^\ell_\mathit{enter}$ be the entry states.
  For every $q_e \in Q^\ell_\mathit{enter}$,
  apply the subset construction wrt.\ $\delta^\ell_N$ starting from $q_e$,
  yielding the following \DFA $\mc F^\ell = (\Sigma, Q^\ell\!,\delta^\ell\!,f^\ell)$.
  The set $Q^\ell \subseteq Q^\ellmo \x 2^V$ of states and $\delta^\ell$ are minimal satisfying:
  \li
  \- $Q^\ell_\mathit{enter} \subseteq Q^\ell$.
  \- Given $(q^\ellmo,V') \in Q^\ell$ and $x \in \Sigma$,
     define
     $\delta^\ell\big((q^\ellmo,V'),x\big) = (q^\ellmo_\mathit{next}, V_\mathit{next})$,
     where
     $\{q^\ellmo_\mathit{next}\} \x V_\mathit{next} = \delta^\ell_N\big(\{(q^\ellmo,v) : v \in V' \},x\big)$.
  \- If $q \in Q^\ell$,
     then $Q^\ell$ also contains all its successors wrt. $\delta^\ell$.
  \il
  Finally,
  the labelling is preserved:
  $f^\ell((q^\ellmo,V')) = f^\ellmo(q^\ellmo)$ for all $(q^\ellmo\!,V') \in Q^\ell$.
\end{definition}

\begin{lemma}\label{lem:NFA-determinization}
  $L(\mc F^\ell) = L(\mc F^\ell_N)$.
\end{lemma}
\begin{proof}
  Follows from the correctness of the subset construction and
  from the fact that $\mc F^\ell$ and $\mc F^\ell_N$ use the labelling functions $f^\ell$ and $f_N^\ell$
  that depend only on the $\mc F^\ellmo$ component.
\end{proof}

\begin{remark}
  The automaton $\mc F^\ell$ may contain transient states and transitions
  (which are not part of any SCC).
  Although such states and transitions can be removed without changing the language of the \DFA,
  this step is omitted in the definition,
  because the LDL-to-DPA translation of the next section
  depends on $\mc F^\ell$ having entry states (which could be outside of any SCC).
  The LDL-to-COCOA translation of this section does not require \DFAs to have special entry states,
  so the removal can be done.
\end{remark}

Figure~\ref{fig:LDLtoDPA:example2:F1} shows an example of determinization
in which applying Definition~\ref{def:FlN-to-Fl} introduces a single new entry state and keeps the other states.

\smallskip
The following lemma states, in particularly, that there is uniform bound on the size of $\mc F^\ell$,
regardless of the level index $\ell$.

\begin{lemma}\label{lem:Fl-size}\label{lem:Fl-uniform-bound}\label{lem:Fl-time}
  Given a weak alternating automaton $\mc A$ with $n$ states,
  an SLTM for $L(\mc A)$,
  obligation graphs $\mcG$ for $L(\mc A)$ and $\tmcG$ for $\overline{L(\mc A)}$,
  their graph labeling of the SLTM states,
  $\ell \in \{1,\ldots,k\}$ where $k$ is the number of levels in the COCOA for $L(\mc A)$, and
  a \DFA $\mc F^\ellmo$ (when $\ell=1$ the \DFA accepts every word),
  consider the following procedure:
  construct \DFA $\mc F^\ell$
  by first creating an \NFA $\mc F^\ell_N$ using Definition~\ref{def:FlN},
  then determinizing it into $\mc F^\ell$ using Definition~\ref{def:FlN-to-Fl}.
  Assuming that $\mc F^\ellmo$ was constructed in the same way,
  the procedure runs in time $\twoexp$\!\! and $\Dl$ has $\twoexp$\!\! states,
  independently of $\ell$.
\end{lemma}

\begin{proof}
  We first prove the claim on the number of states.
  A state of $\mc F^\ell$ has the shape $(q^\ellmo,V^\ell)$,
  where $q^\ellmo$ is a state of $\mc F^\ellmo$ and $V^\ell \subseteq \lambda^\Gl(f^\ellmo(q^\ellmo))$ is a subset of vertices of the obligation graph $\mc G^\ell$
  ($\mc G^\ell$ is $\mcG$ for even $\ell$ and it is $\tmcG$ for odd $\ell$).
  By unfolding the shape and removing nested parenthesis,
  the shape of a state of $\mc F^\ell$ becomes
  $(s,V^1,\ldots,V^\ell)$,
  where
  $s$ is a state of the SLTM and
  $V^i$ is a subset of vertices of the obligation graph $\mc G^i$, for $i\in \{1,\ldots,\ell\}$.
  The number of SLTM states is $\twoexp$\! by Fact~\ref{fact:SLTM},
  the number of obligation-graph vertices is $\oneexp$\!, and
  the number $\ell$ is $\oneexp$\! by Fact~\ref{fact:cocoa-nof-levels}.
  Hence, the number of states of such form is $\twoexp$\!\!.

  Now let us look at the time complexity of the procedure.
  Constructing $\mc F^\ell_N$ from $\mc F^\ellmo$ takes time polynomial in the sizes of
  $\mc F^\ellmo$, $\mc M$, and $\mc G^\ell$.
  Constructing $\mc F^\ell$ from $\mc F^\ell_N$ takes time polynomial in the size of $\mc F^\ell_N$ and exponential in the size of $\mc G^\ell$.
  Therefore, the total running time of the procedure is $\twoexp$\!\!.
\end{proof}

\subsection{Procedure to translate LDL into COCOA}
\label{sec:LDL-to-COCOA:procedure}

\textbf{Procedure \textsc{LDLtoCOCOA}.}
The procedure takes as input an LDL formula or a weak alternating automaton $\mc A$ and
constructs canonical HD-tNCAs $\mc A^1,\ldots,\mc A^k$ of the COCOA:
\lo
\- If the input is an LDL formula $\varphi$,
   translate it into a weak alternating automaton $\mc A$ using~\cite{DBLP:journals/iandc/FaymonvilleZ17}.
\- Construct an SLTM $\mc M$ for $L(\mc A)$ using the algorithm of \cite{angluin2024constructing}
   adapted to weak alternating automata and described in the proof of Fact~\ref{fact:SLTM}.
   Construct obligation graphs for $L(\mc A)$ and $\overline{L({\mc A})}$,
   and the graph labelling functions,
   as described in Section~\ref{sec:graph-labelling}.
\- Construct \DFA $\mc F^0$ wrt.\ $\mc M$ that accepts every word.
\- For $\ell = 1,2,\ldots$, do the following:
   \lo
   \- Construct the \NFA $\mc F^\ell_N$ using Definition~\ref{def:FlN}
      from the previously computed $\mc F^\ellmo$.
   \- Determinize the \NFA $\mc F^\ell_N$ into \DFA ${\mc F^\ell}$ using Definition~\ref{def:FlN-to-Fl}.
   \- Translate $\mc F^\ell$ into an HD-tNCA using the construction from the proof of Lemma~\ref{lem:DFA-to-HD-tNCA}.
   \- Minimize-canonize the HD-tNCA using the construction of~\cite{DBLP:journals/lmcs/RadiK22},
      yielding $\mc A^\ell$.
   \ol
   The procedure stops when $L(\mc A^\ell)=\emptyset$ and
   returns $(\mc A^1,\ldots,\mc A^\ellmo)$,
   or returns $()$ when $L(\mc A^1)=\emptyset$.
\ol

\begin{theorem}
  Procedure \textsc{LDLtoCOCOA} constructs the COCOA $(\mc A^1,\ldots,\mc A^k)$
  from either an LDL formula or a weak alternating automaton in time $\twoexp\!$,
  where $n$ is the size of the formula when the input is a formula and otherwise
  it is the number of automaton states.
\end{theorem}

\begin{proof}
  The correctness of the constructed COCOA (that the computed HD-tNCAs indeed reprsenent the COCOA levels) follows from Lemmas~\ref{lem:Ll_includes_AlN}, \ref{lem:AlN_includes_Ll}, \ref{lem:NFA-determinization}, and \ref{lem:DFA-to-HD-tNCA}.
  The time complexity is implied by the following.
  The first step takes time $2^{\mathit{poly(|\varphi|)}}$ by Fact~\ref{fact:LDL-to-Weak} and produces an automaton with $n \in O(|\varphi|)$ states.
  The second step takes time $\twoexp$\! by Fact~\ref{fact:SLTM}, Lemma~\ref{lem:graph-labelling}, and the $2^{O(n)}$ time complexity of the Miyano-Hayashi construction~\cite{MHConstruction}.
  The third step takes time linear in the size of the SLTM, i.e., $\twoexp$\!.
  The fourth step takes time $\twoexp$ since there are $\oneexp$ COCOA levels (Fact~\ref{fact:cocoa-nof-levels}) and each iteration takes $\twoexp$\! time, independently of the value $\ell$,
  which follows from Lemma~\ref{lem:Fl-uniform-bound},
  Lemma~\ref{lem:DFA-to-HD-tNCA},
  and Fact~\ref{fact:min-can-HD-tNCA}.
\end{proof}

\subsection{Lower bound for translating LDL into COCOA}

\newcommand\propCOCOAlowerBoundStatement{%
  Translating LTL into COCOA requires at least doubly exponential time.
  Specifically,
  there exist LTL formulas $\Phi_1,\Phi_2,\ldots$ that satisfy
  the following conditions.
  For every number $n$:
  \li
  \- the formula $\Phi_n$ has size $O(n)$ and is is defined over the set of APs of size $O(n)$;
  \- the COCOA for $L(\Phi_n)$ has size at least $2^{2^n}\!\!$,
     where the size of a COCOA is defined as the total number of states on all levels.
  \il
}
\begin{proposition}
  \label{prop:COCOA-lower-bound}
  \propCOCOAlowerBoundStatement
\end{proposition}
The proof relies on the result in~\cite{KR11},
which describes a family of LTL formulas $\Psi_1,\Psi_2,\ldots$ of stated sizes
such that a minimal DCA recognizing $L(\Psi_n)$ has doubly exponentially many states.
These DCAs are in fact among the smallest HD-tNCAs,
because they satisfy the property of the minimality stated in \cite{DBLP:journals/lmcs/RadiK22}.
Therefore, the corresponding HD-tNCAs also have a doubly exponential number of states.
The COCOA for $L(\neg \Psi_n)$ contains a single automaton with doubly exponentially many states.

\section{Translating LDL into tDPA}
\label{sec:LDL-to-DPA}

\subsection{High-level idea}
The procedure \textsc{LDLtoCOCOA} from the previous section
internally constructs a chain of \DFAs $(\mc F^1,\ldots,\mc F^k)$
representing the languages of the levels of the COCOA for $L(\varphi)$,
where $\varphi$ is a given LDL formula.
The procedure \textsc{LDLtoDPA} presented in this section
uses those \DFAs (except, perhaps, $\mc F^0$) to construct a tDPA recognizing $L(\varphi)$.
Intuitively,
the constructed tDPA simulates \DFAs from all levels and tries to stay at the deepest level.
The state space of the tDPA consists of states of the \DFAs.
Initially,
the tDPA starts in the deepest entry state labelled by the initial state of the SLTM.
Let the current tDPA state be a state $q^\ell$ of the \DFA $\mc F^\ell$.
On reading a letter,
the tDPA increases the level, decreases the level, or stays on the same level,
depending on the following.
\li
\- The state $q^\ell$ does not have a successor in $\Dl$ :
   the tDPA decreases the level until it finds $\Dlp$ that has a successor state,
   and transits into that state.
   The tDPA can find such $\Dlp$
   because the current state $q^\ell$ encodes the states of all lower-level \DFAs as well.
   (Recall that $q^\ell$ has shape $(q^\ellmo,V)$, where $q^\ellmo$ is a state of the \DFA $\mc F^\ellmo$.)
\- The state $q^\ell$ has a successor in $\Dl$:
   if the successor also has a refining state in a \DFA at a deeper level,
   then the tDPA transits into the entry state of the deepest such \DFA.
   Otherwise,
   the tDPA transits into the successor state of $q^\ell$ in $\Dl$.
\il
The color of the tDPA transition is equal to the smallest level participating in this dance.

\subsection{Procedure to translate LDL into tDPA}
\newcommand\penter{\mathit{pstate}}

The translation procedure consists of two functions,
the entry function \textsc{LDLtoDPA} (Figure~\ref{fig:LDL-to-DPA-entry}) and the recursive function \textsc{RecBuild} (Figure~\ref{fig:LDL-to-DPA-rec}).

\parit{\textsc{LDLtoDPA}}.
The function takes an LDL formula $\varphi$ and returns a tDPA recognizing $L(\varphi)$.
It starts by constructing
a weak alternating automaton $\mc A$ for $\varphi$ \cite{DBLP:journals/iandc/FaymonvilleZ17},
an SLTM $\mc M$ for $L(\mc A)$ (proof of Fact~\ref{fact:SLTM}), and
obligation graphs $\mc G_\mc A$, $\mc G_{\bar{\mc A}}$ with their graph labelling functions
(Section~\ref{sec:graph-labelling}).
It then builds \DFA $\mc F^0$ by re-using the whole transition structure of the SLTM (including transient states and transitions)  and
by using the identity function for the state labelling function of $\mc F^0$.
It initiates the recursive procedure by calling \textsc{RecBuild}($0,\mc F^0$),
which also has access to $\mc M$, $\mc G_{\mc A}$, $\mc G_{\bar{\mc A}}$, $\lambda^{\mcG_{\mc A}}$, $\lambda^{\mcG_{\bar{\mc A}}}$.
\textsc{RecBuild}($0,\mc F^0$) returns a tDPA structure $\mc P^0$ (with no initial state) and
a mapping $\penter^0$ that maps each state of $\mc F^0$ to a state of $\mc P^0$.
Finally, we return $\mc P^0$ with the initial state $\penter^0(s_0)$,
where $s_0$ is the initial state of the SLTM.

\begin{figure}[t]
  \begin{algorithmic}[1]
    \algrenewtext{Function}[2]{\algorithmicfunction\ \textproc{#1}(#2):}
    \algrenewcommand\algorithmicthen{:}
    \algrenewcommand\algorithmicdo{:}
    \algrenewcommand\algorithmicelse{\textbf{else:}}
    \Function{LDLtoDPA}{$\varphi$}
      \State translate $\varphi$ into a weak alternating automaton $\mc A$
      \State construct an SLTM $\mc M$ for $L(\mc A)$, oblig.\ graphs $\mc G_{\mc A}$, $\mc G_{\bar{\mc A}}$, the graph labelling $\lambda^{\mc G_{\mc A}}$\!,\,$\lambda^{\mc G_{\bar{\mc A}}}$
      \State derive the trivial \DFA $\mc F^0$ from the SLTM $\mc M$
      \State $\penter^0\!, \mc P^0 = \textsc{RecBuild}(0,\mc F^0)$
      \State \Return $\mc P^0$ with the initial state being $\penter^0(s_0)$,
                     where $s_0$ is the initial state of $\mc M$
    \EndFunction
  \end{algorithmic}
  \caption{The entry function for translating LDL into tDPA.}
  \label{fig:LDL-to-DPA-entry}
\end{figure}
\begin{figure}[t]
  \begin{algorithmic}[1]
    \algrenewtext{Function}[2]{\algorithmicfunction\ \textproc{#1}(#2):}
    \algrenewcommand\algorithmicthen{:}
    \algrenewcommand\algorithmicdo{:}
    \algrenewcommand\algorithmicelse{\textbf{else:}}
    \Function{RecBuild}{$\ell$, $\Dl$}
      \State \label{l2}initialize $\penter^\ell$ and $\mc P^\ell$ as empty
      \State \label{l3}\textbf{if} $\Dl$ has no states: \Return{$\penter^\ell\!,\, \mc P^\ell$}
      \vspace{1.2mm}
      \State \label{l4}\label{alg:line:NFA}construct \NFA $\Dll_N$ from $\Dl$, $\mc M$, $\mc G_{\mc A}$, $\mc G_{\bar{\mc A}}$ using Definition~\ref{def:FlN}
      \State \label{l5}\label{alg:line:det}determinize $\Dll_N$ into \DFA $\Dll$ with entry states using Definition~\ref{def:FlN-to-Fl}
      \State $\penter^\lPlus\!,\,\mc P^\lPlus \gets \Call{RecBuild}{\lPlus, \Dll}$\label{l6}
      \State add $\mc P^\lPlus$ to $\mc P^\ell$\label{l7}
      \vspace{1.2mm}
      \For{every state $q^\ell$ of $\Dl$}\label{l8}
        \If{$q^\ell$ has a refining state in $\Dll$}\label{l9}
          \State\label{alg:line:pstate-update}$\penter^\ell\!\!: q^\ell \mapsto \penter^\lPlus((q^\ell\!, \Vinit_{q^\ell}))$,\, where $(q^\ell\!,\Vinit_{q^\ell})$ is an entry state of $\Dll$\!\!\!\!\!\!\!\label{l10}
        \Else
          \State $\penter^\ell\!\!:q^\ell \mapsto q^\ell$\label{l12}
          \State add state $q^\ell$ to $\mc P^\ell$\label{alg:line:add-state}\label{l13}
        \EndIf
      \EndFor
      \vspace{1.2mm}

      \For{every transition $q^\ell \trans{x} {q^\ell}'$ of $\Dl$ that has no refining transitions in $\Dll$}\label{alg:line:corr-trans}\label{l14}\ak{TODO: BUG!!!}
        \State add transition $\penter^\ell(q^\ell) \trans{x} \penter^\ell({q^\ell}')$ with color $\ell$ to $\mc P^\ell$\label{alg:line:add-trans}\label{l15}
      \EndFor
      \vspace{1mm}

      \State \Return{$\penter^\ell\!,\, \mc P^\ell$}\label{alg:line:return}\label{l16}
    \EndFunction
  \end{algorithmic}
  \caption{The key function for translating LDL into tDPA.
           The function has access to $\mc M$, $\mc G_{\mc A}$, $\mc G_{\bar{\mc A}}$,
           $\lambda^{\mc G_{\mc A}}$\!, $\lambda^{\mc G_{\bar{\mc A}}}$.}
  \label{fig:LDL-to-DPA-rec}
\end{figure}

\parit{\textsc{RecBuild}}.
The function accepts a number $\ell$ representing the current level and a \DFA $\mc F^\ell$.
At the beginning,
the tDPA structure $\mc P^\ell$ (a tDPA without an initial state) has no states and transitions
and the mapping $\penter^\ell$ is empty.
The purpose of $\penter^\ell$ is to map each state $q^\ell$ of $\Dl$ to the entry state of the deepest $\Dlp$ that refines\footnote{A state $q^{\ell'}$ of $\Dlp$ \emph{refines} state $q^\ell$ of $\Dl$,
            where $\ell'>\ell$,
            if the projection of $q^{\ell'}$ on the states of $\Dl$ is $q^\ell$.}
$q^\ell$.
In lines \ref{l4}--\ref{l5},
we construct \DFA $\Dll$ from $\Dl$ using Definitions~\ref{def:FlN}\,and\,\ref{def:FlN-to-Fl}.
We then recurse onto deeper level, and
add all states and transitions constructed on deeper levels to the current $\mc P^\ell$.
Then,
each state $q^\ell$ of $\Dl$ that has a refining state in $\Dll$ is mapped to $\penter^\lPlus((q^\ell,\Vinit_{q^\ell}))$,
where $(q^\ell,\Vinit_{q^\ell})$ is an entry state of $\Dll$.
(Recall the purpose of entry states:
 a state $q^\ell$ of $\Dl$ can have several refining states in the \DFA $\Dll$,
 so the entry state $(q^\ell,\Vinit_{q^\ell})$ merges them into a single one,
 preserving their paths and avoiding the need to choose between them.)
Note that the state $(q^\ell,\Vinit_{q^\ell})$ into which we map $q^\ell$ may itself be mapped to a state in a deeper \DFA.
Lines \ref{l12}--\ref{l13}: If a state $q^\ell$ has no refining states at a deeper level,
we map it to itself, and add this state to $\mc P^\ell$.
Lines \ref{l14}--\ref{l15} add a transition of $\Dl$ to $\mc P^\ell$ if there is not yet a refining transition at a deeper level.
Notice that both states --
the source $\penter^\ell(q^\ell)$ and the destination $\penter^\ell({q^\ell}')$ --
may refer to states at a deeper level.

\subsection{Examples}

\begin{example}
  We run \textsc{LDLtoDPA} for the formula
  $\varphi = \LTLG(a \impl \LTLX \neg a) \land (\LTLGF a \impl \LTLGF b)$
  over the alphabet $\Sigma=2^{\{a,b\}}$.
  An SLTM for $L(\varphi)$ is in Figure~\ref{fig:SLTM}, obligation graphs are in Figures~\ref{fig:NBA}~and~\ref{fig:NBA_neg}.
  The \DFAs constructed by \textsc{LDLtoDPA} are very similar
  to those in Figure~\ref{fig:FlN}.
  The difference is that the state $(s_0,\tilde v_3, v_2)$ of $\mc F^2_N$
  becomes $(s_0,\{\tilde v_3\}, \{v_2\})$ in the \DFA $\mc F^2$,
  because we always execute to the determinization step in Line~\ref{l5}.
  Determinization also marks every state of $\mc F^1$ and $\mc F^2$ as an entry state.
  For instance,
  $(s_0,\{\tilde v_3\})$ is an entry state of $\mc F^1$
  because the \NFA $\mc F^1_N$ contains the only state, namely $(s_0,\tilde v_3)$,
  whose parent-\DFA component is $s_0$.
  Similarly,
  $(s_0,\{\tilde v_3\},\{v_2\})$ is an entry state of $\mc F^2$
  because $\mc F^2_N$ has just one state with the same parent-\DFA component $(s_0,\tilde v_3)$.


  The functions $\penter^\ell$ jointly induce the following chains,
  where $x \mapsto_\ell y$ means $\penter^\ell(x) = y$, for $\ell\in\{0,1,2\}$:
  \li
  \- $s_0 \mapsto_0 (s_0,\{\tilde v_3\}) \mapsto_1 (s_0,\{\tilde v_3\},\{v_2\}) \mapsto_2 (s_0,\{\tilde v_3\},\{v_2\})$,
  \- $s_1 \mapsto_0 (s_1,\{\tilde v_3\}) \mapsto_1 (s_1,\{\tilde v_3\})$, and
  \- $s_2 \mapsto_0 (s_2,\{\tilde v_2\}) \mapsto_1 (s_2,\{\tilde v_2\})$.
  \il
  The terminal states in these chains form the states of the resulting tDPA.
  In \textsc{RecBuild($2,\mc F^2$)},
  Line~\ref{l13} adds the state $(s_2,\{\tilde v_2\})$ to $\mc P^2$.
  Line~\ref{l15} then adds a self-loop $\bar a \bar b$ with color 2.
  No futher states or transitions are added at this level.
  In \textsc{RecBuild($1,\mc F^1$)},
  the states $(s_1,\{\tilde v_3\})$ and $(s_2,\{\tilde v_2\})$ are added to $\mc P^1$
  (which at this point includes $\mc P^2$).
  Since the transition $(s_0,\{\tilde v_3\})\trans{a\bar b} (s_1,\{\tilde v_3\})$ of $\mc F^1$
  has no refining transition in $\mc F^2$,
  Line~\ref{l15} adds the transition $(s_0,\{\tilde v_3\},\{v_2\}) \trans{a\bar b} (s_1,\{\tilde v_3\})$ with color 1 to $\mc P^1$.
  And so on.
  Finally,
  in \textsc{RecBuild($0,\mc F^0$)},
  the remaining transitions are added to make the tDPA complete (thanks to $\mc F^0$ being complete).
  In particular,
  we add these transitions:
  $(s_0,\{\tilde v_3\},\{v_2\}) \trans{\bar a b} (s_0,\{\tilde v_3\},\{v_2\})$ with color 0 and
  $(s_1,\{\tilde v_3\}) \trans{a} (s_2,\{\tilde v_2\})$ with color 0.
  Figure~\ref{fig:LDLtoDPA:example} shows the resulting tDPA.
\end{example}

\begin{figure}
  \centering
  \begin{tikzpicture}[->,>=stealth', auto, node distance=15mm,every node/.style={font=\footnotesize}]
    \tikzstyle{every state}=[rectangle,rounded corners,text=black,font=\footnotesize,initial text={},minimum size=4mm, inner sep=0.6mm]
    \tikzset{
      every edge/.append style = {font=\footnotesize}
    }
    \tikzstyle{marked}=[
      draw=green!38!black,
      postaction={
        decorate,
        decoration={
          markings,
          mark=at position 0.5 with {
            \fill[green!38!black] (0,0) circle[radius=0.4mm];
          }
        }
      }
    ]

    \node[state,initial] (p0) {$(s_0,\{\tilde v_3\},\{v_2\})$};
    \node[state] (p1) [right=1.4cm of p0] {$(s_1,\{\tilde v_3\})$};
    \node[state] (p2) [right=1cm of p1] {$(s_2,\{\tilde v_2\})$};

    \path[->] (p0) edge [loop,out=75,in=105,looseness=10] node[above,yshift=0mm] {$\bar a\bar b/2$} (p0)
              (p0) edge [loop,in=255,out=285,looseness=10] node[below,yshift=0mm] {$\bar a b/0$} (p0)
              (p0) edge [bend right=4] node[below,yshift=0.8mm]{$a \bar b / 1$} (p1)
              (p0) edge [bend right=30] node[below,yshift=0.8mm]{$a b / 0$} (p1)
              (p1) edge [bend right=4] node[above,yshift=-0.8mm]{$\bar a\bar b / 1$} (p0)
              (p1) edge [bend right=30] node[above,yshift=-0.8mm]{$\bar a b / 0$} (p0)
              (p1) edge [] node[below,yshift=0.8mm]{$a / 0$} (p2)
              (p2) edge [loop,out=75,in=105,looseness=10] node[above,xshift=0mm,yshift=0mm] {$*/1$} (p2)
              ;

  \end{tikzpicture}
  \caption{tDPA for the language of $\varphi = \LTLG(a \impl \LTLX \neg a) \land (\LTLGF a \impl \LTLGF b)$
           constructed by \textsc{LDLtoDPA} from \NFAs
           in Figure~\ref{fig:FlN}.}
  \label{fig:LDLtoDPA:example}
\end{figure}

\begin{example}
  This example illustrates the importance of entry states introduced during determinization in Definition~\ref{def:FlN-to-Fl}.
  Consider the language of $\varphi = \LTLGF(a\land \LTLX a)$ over the alphabet $\Sigma=2^{\{a\}}$.
  Its SLTM consists of a single state $s_0$ with a self-loop on every letter.
  The obligation graphs are shown in Figure~\ref{fig:LDLtoDPA:example2:graphs}.
  The \NFA $\mc F^1_N$ and the \DFA $\mc F^1$ are shown in Figure~\ref{fig:LDLtoDPA:example2:F1}.
  Notice that determinization introduces the state $(s_0,\{\tilde v_1,\tilde v_2\})$,
  which is an entry state.
  Removing this state does not change the language of the \DFA.
  However,
  this state is important for \textsc{LDLtoDPA} because it resolves the choice between the two states of the \NFA.
  The resulting tDPA is Figure~\ref{fig:LDLtoDPA:example2:F1}.
\end{example}

\begin{figure}[t]
  \centering
  \begin{tikzpicture}[->,>=stealth', auto, node distance=15mm]
    \tikzstyle{every state}=[text=black,font=\footnotesize,initial text={},minimum size=5mm, inner sep=0mm]
    \tikzset{
      every edge/.append style = {font=\footnotesize}
    }
    \tikzstyle{marked}=[
      draw=green!38!black,
      postaction={
        decorate,
        decoration={
          markings,
          mark=at position 0.5 with {
            \fill[green!38!black] (0,0) circle[radius=0.4mm];
          }
        }
      }
    ]

    \node[state,initial] (v0) {$v_0$};
    \node[state] (v1) [right=7mm of v0] {$v_1$};

    \path[->] (v0) edge [loop,out=75,in=105,looseness=8] node[above,xshift=0mm,yshift=0mm] {$*$} (v0)
              (v0) edge [bend right=10] node[below]{$a$} (v1)
              (v1) edge [marked,bend right=10] node[above,yshift=0mm]  {$a$} (v0);

    \node[state,initial left] (nv0) [right=1.8cm of v1] {$\tilde v_0$};
    \node[state] (nv1) [right=5mm of nv0] {$\tilde v_1$};
    \node[state] (nv2) [right=7mm of nv1] {$\tilde v_2$};

    \path[->]
         (nv0) edge [loop,out=75,in=105,looseness=8] node[above,xshift=0mm,yshift=0mm] {$*$} (nv0)
         (nv0) edge node[below,xshift=0mm,yshift=0mm] {$*$} (nv1)
         (nv1) edge[bend right=10] node[below,xshift=0mm,yshift=0mm] {$a$} (nv2)
         (nv2) edge[marked, bend right=10] node[above,xshift=0mm,yshift=0mm] {$\bar a$} (nv1)
         (nv1) edge [marked,loop,out=75,in=105,looseness=8] node[above,xshift=0mm,yshift=0mm] {$\bar a$} (nv1)
    ;

  \end{tikzpicture}
  \caption{On the left side:
           a tNBA $\mc G_\varphi$ recognizing the language of $\varphi=\LTLGF(a\land \LTLX a)$.
           On the right side:
           a tNBA $\mc G_{\neg\varphi}$ recognizing $L(\neg\varphi)$.
           The SLTM contains a single state $s_0$ with the self-loop.
           The graph labelling are
           $\lambda^{\mc G_\varphi}(s_0) = \{v_0,v_1\}$ and
           $\lambda^{\mc G_{\neg\varphi}}(s_0) = \{\tilde v_0,\tilde v_1,\tilde v_2\}$.
           }
  \label{fig:LDLtoDPA:example2:graphs}
\end{figure}

\begin{figure}[t]
  \centering
  \begin{tikzpicture}[->,>=stealth', auto, node distance=15mm,every node/.style={font=\footnotesize}]
    \tikzstyle{every state}=[rectangle,rounded corners,text=black,font=\footnotesize,initial text={},minimum size=4mm, inner sep=0.6mm]
    \tikzset{
      every edge/.append style = {font=\footnotesize}
    }
    \tikzstyle{marked}=[
      draw=green!38!black,
      postaction={
        decorate,
        decoration={
          markings,
          mark=at position 0.5 with {
            \fill[green!38!black] (0,0) circle[radius=0.4mm];
          }
        }
      }
    ]

    \node[state] (nv1) {$(s_0,\tilde v_1)$};
    \node[state] (nv2) [right=5mm of nv1] {$(s_0,\tilde v_2)$};

    \path[->]
         (nv1) edge[bend right=10] node[below,xshift=0mm,yshift=0mm] {$a$} (nv2)
         (nv2) edge[bend right=10] node[above,xshift=0mm,yshift=0mm] {$\bar a$} (nv1)
         (nv1) edge [loop,out=75,in=105,looseness=8] node[above,xshift=0mm,yshift=0mm] {$\bar a$} (nv1)
    ;

    \node[state] (dnv1) [right=10mm of nv2]{$(s_0,\{\tilde v_1\})$};
    \node[state] (dnv0) [above=7mm of dnv1, xshift=10mm]{$(s_0,\{\tilde v_1,\tilde v_2\})$};
    \node[state] (dnv2) [right=6mm of dnv1] {$(s_0,\{\tilde v_2\})$};

    \path[->]
         (dnv0) edge[bend right=10] node[right,xshift=0mm,yshift=0mm] {$\bar a$} (dnv1)
         (dnv0) edge[bend left=10] node[right,xshift=0mm,yshift=0mm] {$a$} (dnv2)
         (dnv1) edge[bend right=10] node[below,xshift=0mm,yshift=0mm] {$a$} (dnv2)
         (dnv2) edge[bend right=10] node[above,xshift=0mm,yshift=0mm] {$\bar a$} (dnv1)
         (dnv1) edge [loop,out=75,in=105,looseness=8] node[above,xshift=0mm,yshift=0mm] {$\bar a$} (dnv1)
    ;

    \node[state] (p1) [right=10mm of dnv2]{$(s_0,\{\tilde v_1\})$};
    \node[state,initial left] (p0) [above=8mm of p1, xshift=11mm]{$(s_0,\{\tilde v_1,\tilde v_2\})$};
    \node[state] (p2) [right=9mm of p1] {$(s_0,\{\tilde v_2\})$};

    \path[->]
         (p0) edge[bend right=10] node[left,xshift=1.5mm,yshift=1.2mm] {$\bar a/1$} (p1)
         (p0) edge[bend left=10] node[left,xshift=0mm,yshift=0mm] {$a/1$} (p2)
         (p1) edge[bend right=5] node[below,xshift=0mm,yshift=0.8mm] {$a/1$} (p2)
         (p2) edge[bend right=5] node[above,xshift=0mm,yshift=-0.9mm] {$\bar a/1$} (p1)
         (p1) edge [loop,out=75,in=105,looseness=8] node[left,xshift=0mm,yshift=0mm] {$\bar a/1$} (p1)
         (p2) edge[bend right=20,out=300,in=210] node[right,xshift=0mm,yshift=0mm] {$a/0$} (p0)
    ;

  \end{tikzpicture}
  \caption{From left to right: \NFA $\mc F^1_N$, \DFA $\mc F^1$\!,
           and tDPA computed by \textsc{LDLtoDPA} for $\varphi = \LTLGF(a \land \LTLX a)$ and the NBAs in Figure~\ref{fig:LDLtoDPA:example2:graphs}.}
  \label{fig:LDLtoDPA:example2:F1}
\end{figure}

\subsection{Procedure correctness}

Before proving the correctness, we introduce two concepts:
the projection of tDPA states into \DFA state space and
a notion of a word ``reaching'' a \DFA.

Every state of a tDPA constructed by the algorithm is a state of some \DFA.
Consider the shape of \DFA states.
The states of the \DFA $\mc F^0$ are exactly the states of the SLTM $\mc M$.
A state of a \DFA from level $\ell\geq1$ has the form $(q,V^\ell)$,
where
$q$ is a state of its parent $\mc F^{\ell-1}$ and $V^\ell$ is a set of vertices of the obligation graph $\mc G^\ell$,
and $\mc G^\ell = \mc G_\mc A$ if $\ell$ is even and $\mc G^\ell = \mc G_{\bar{\mc A}}$ if $\ell$ is odd.
By unfolding the recursive definition of \DFA states and omitting nested parentheses,
a state of a \DFA at level $\ell$ can be written as $(s,V^1,\ldots,V^\ell)$,
where
$s$ is a state of the SLTM and $V^i$ is a set of vertices of the obligation graph $\mc G^i$,
for all $i \in \{1,\ldots,\ell\}$.
Given a tDPA state $q = (s,V^1,\ldots,V^\ell)$ and a \DFA $\mc F^i$ at level $i$,
the projection $q_{\mid \mc F^i}$ is $(s,V^1,\ldots,V^i)$
when $i\leq\ell$ and the state $(s,V^1,\ldots,V^i)$ belongs to $\mc F^i$,
otherwise $q_{\mc F^i} = \circ$,
where $\circ$ is a distinguished symbol.
Given a tDPA transition $(q,x,q')$,
its projection $(q,x,q')_{\mid \mc F^i}$ onto \DFA $\mc F^i$ is
$(q_{\mid \mc F^i},x,q'_{\mid \mc F^i})$.
Note that $(q_{\mid \mc F^i},x,q'_{\mid \mc F^i})$ might be not a valid transition of $\mc F^i$,
even when both $q_{\mid \mc F^i}$ and $q'_{\mid \mc F^i}$ are not $\circ$.

A tDPA run $(q_0,x_0,q_1)(q_1,x_1,q_2)\ldots$
\emph{enters} a \DFA $\mc F$ at moment $j \geq 0$
if
${q_{j}}_{\mid \mc F} \neq \circ$ and
either $j=0$, or for $j>0$,
$\big({q_{j-1}}_{\mid \mc F},x_{j-1},{q_{j}}_{\mid \mc F}\big)$
is \emph{not} a transition of $\mc F$.

\begin{lemma}\label{lem:dpa-enters-init}
  A run of the tDPA can enter a \DFA $\Dl$ of a level $\ell$ only through an entry state of $\Dl$.
\end{lemma}

\begin{proof}
  \newcommand\projDn{{\mid \mc F^n}}
  Let $j$ be the moment when the run enters $\Dl$.
  There are two cases: $j>0$ and $j=0$.

  Suppose $j>0$.
  Let $(q_\mathit{prev}, x, q)$ be the transition of the run taken at moment $j-1$.
  Since $({q_\mathit{prev}}_\projDl, x, q_\projDl) \not\in\delta^{\Dl}$,
  there exists $n<\ell$ (not necessary $\ell-1$)
  such that
  the guard in Line~\ref{l14} of \textsc{RecBuild($n,\mc F^n$)} holds true:
  $({q_\mathit{prev}}_\projDn,x,q_\projDn) \in \delta^{\mc F^n}$ and there are no refining transitions.
  Then, according to Line~\ref{l14}, the tDPA has the transition
  $\penter^n({q_\mathit{prev}}_\projDn) \trans{x} \penter^n(q_\projDn)$,
  which is taken at moment $j-1$ of the run.
  Let us see how $\penter^n({q}_\projDn)$ was defined during the execution of \textsc{RecBuild}.
  Let $q^i = q_{\mid \mc F^i}$ for $i = 0,\ldots,\ell$.
  Observe that Line~\ref{l12} was not executed
  on any of the states $q^0,\ldots,q^{\ell-1}$,
  because they all have a refining state $q^\ell$.
  For these states,
  the guard on Line~\ref{l9} is true,
  and Line~\ref{l10} is executed.
  On level $\ell-1$,
  the state $q^{\ell-1}$ is mapped to $\penter^\ell((q^{\ell-1},\Vinit_{q^{\ell-1}}))$.
  These mappings are chained:
  $q = \penter^0(q^0) = \ldots = \penter^\ell((q^{\ell-1},\Vinit_{q^{\ell-1}}))$.
  Thus,
  $q_\projDl = \penter^\ell((q^{\ell-1},\Vinit_{q^{\ell-1}}))_\projDl = (q^{\ell-1},\Vinit_{q^{\ell-1}})$.
  The latter state is an entry state of $\Dl$.

  The proof of the case $j=0$ is similar and omitted.
\end{proof}

We define what it means for a word to \emph{reach} a \DFA.
Given a word $w = x_0 x_1 \ldots$,
let $\rho = q_0 q_1 q_2\ldots$ be the corresponding state run of the tDPA.
For a moment $m$,
let $\rho[m{:}]$ denote the suffix of $\rho$ starting at $m$,
and let $w[m{:}]$ denote the suffix of $w$ starting at $m$.
For a \DFA $\Dl$ at level $\ell$,
let $\rho[m{:}]_\projDl$ denote the projection of the suffix run onto $\Dl$.
Observe that $\rho[m{:}]_\projDl$ is not necessarily a path of $\Dl$ on $w[m{:}]$:
there might exist $j\geq m$ s.t.\ $\rho[j]_\projDl = \circ$, or
there might exist $j$ s.t.\ $({q_j}_\projDl,x_j,{q_{j+1}}_\projDl)$ is not a transition of $\Dl$
(even if there are no $\circ$).
When
there exists a moment $m$
such that the projection $\rho[m{:}]_{\projDl}$ is an infinite path on $w[m{:}]$ starting in an entry state of $\Dl$,
we say that a word $w$ \emph{reaches} $\Dl$.
This means that the run $\rho$ of the tDPA on $w$ gets trapped in an SCC of $\Dl$,
and possibly in a deeper one.

\begin{lemma}\label{lem:ldl-to-dpa}
  Fix $\ell\geq0$, \DFA $\Dl$ from level $\ell$, and a word $w$ that reaches $\Dl$.
  Then
  \li
  \- for odd $\ell$: $w \in L \Rightarrow$ $w$ reaches $\Dll$;
  \- for even $\ell$: $w \not\in L \Rightarrow$ $w$ reaches $\Dll$.
  \il
\end{lemma}

\begin{proof}
  Consider the case of odd $\ell$; the case of even $\ell$ is similar.

  Let $\rho$ be the state run of the tDPA on $w$.
  Let $m_1$ be the moment,
  after which every state that $\rho$ visits,
  appears infinitely often.
  Since $w$ reaches $\Dl$,
  the projection $\rho[m_1{:}]_\projDl$
  is an infinite path of $\Dl$ on $w[m_1{:}]$,
  and every state appearing in $\rho[m_1{:}]_\projDl$
  appears there infinitely often.

  Since $w \in L$,
  it induces an accepting run $\gamma = v_0 v_1 \ldots$ in the obligation graph $\mc G_{\mc A}$.
  After some moment $m_2$,
  the run $\gamma$ starts visiting only states that appear infinitely often in it.

  Let $m'$ be the maximal between $m_1$ and $m_2$.

  We combine $\rho[m'{:}]_\projDl$ and $\gamma[m'{:}]$
  into the product sequence
  $(q_{m'},v_{m'})(q_{m'+1},v_{m'+1})\ldots$,
  where $q_j = \rho[j]_\projDl$ and $v_j = \gamma[j]$ for all $j\geq m'$.
  Let $m\geq m'$ be the moment after which this product sequence
  visits only pairs that appear infinitely often,
  and let $\pi = (q_m,v_m)(q_{m+1},v_{m+1})\ldots$ be the corresponding product.

  \medskip
  \noindent
  \emph{Claim 1:
        the product $\pi$ is an infinite path of the \NFA $\mc F_N^\lPlus$ on $w[m{:}]$.
        }\\
  Proof:
  The claim follows from Definition~\ref{def:FlN}.
  First,
  in the pair $(q_m,v_m)$,
  the vertex $v_m$ belongs to the SLTM label of $q_m$
  \footnote{Recall that the very first component of the state $q_m$ of $\Dl$ is a state of the SLTM,
            which is used as the definition of the label of $q_m$ when applying Definition~\ref{def:FlN}.},
  because the vertex $v_m$ and the state $q_m$ (hence its label) are reachable on the same prefix $w[0{:}m{-}1]$.
  Thus,
  in Step 1 of Definition~\ref{def:FlN},
  $(q_m,v_m)$ is added to the preliminary transition structure of the \NFA $\mc F^\lPlus_N$.
  This state together with all subsequent transitions of $\pi$
  belong to some SCC of the preliminary transition structure.
  Moreover, $\pi$ visits a $\mc G_{\mc A}$-accepting transition infinitely often because $\gamma$ does.
  Therefore,
  Step 2 of Definition~\ref{def:FlN} declares that the SCC containing $\pi$ belongs to the finally built \NFA $\mc F^\lPlus_N$.

  \medskip
  \noindent
  \emph{Claim 2:
        for every state $q$ of $\Dl$ in $\pi$,
        the \DFA $\Dll$ has a unique entry state.}\\
  Proof:
  By Claim~1,
  for every such $q$
  there is at least one vertex $v$ s.t.\ the \NFA $\mc F^\lPlus_N$ contains $(q,v)$.
  The claim then follows from the definition of entry states of \DFA $\Dll$.

  \medskip
  \noindent
  \emph{Claim 3:
        for each $(q,v)$ of $\pi$,
        the vertex $v$ belongs to $\Vinit_q$,
        where $(q,\Vinit_q)$ is an entry state of $\Dll$.}\\
  Proof:
  By Claim~1,
  each $(q,v)$ of $\pi$ belongs to $\mc F^\lPlus_N$.
  By Claim~2,
  the state $q$ has an entry state $(q,\Vinit_q)$ in $\Dll$.
  By the definition of entry states of $\Dll$,
  the set $V$ contains every $v$ s.t.\ $(q,v)$ belongs to the \NFA $\mc F^\lPlus_N$,
  implying the claim.

  \medskip
  \noindent
  \emph{Claim 4:
        Fix a starting moment $j \geq m$;
        let $(q_j,v_j)$ be the corresponding state in $\pi$;
        let $(q_j,\Vinit_j)$ be the corresponding entry state of $\Dll$.
        Then:
        the suffix $w[j{:}]$ has an infinite path in $\Dll$ starting from $(q_j,\Vinit_j)$,
        and $v_j \in \Vinit_j$.}\\
  Proof:
  Let $\rho^\star = (q_j,\Vinit_j)(q_{j+1},V_{j+1})\ldots$ be the maximal path of $\Dll$ on $w[j{:}]$.
  We need to show that $\rho^\star$ has infinite length.
  The infinity of $\rho^\star$ is implied by the infinity of $\pi$ and the following statement.
  For each moment $n\geq j$,
  the state $(q_n,V_n)$ of $\rho^\star$ ensures that $v_n \in V_n$.
  Let us prove this.
  Initially,
  $v_j \in \Vinit_j$ by Claim~3.
  Assume this for $n\geq j$, so $v_n \in V_n$,
  and consider the successor state $(q_{n+1},V_{n+1}) = \delta^{\Dll}((q_n,V_n),w[n])$.
  By Claim~1, $(q_{n+1},v_{n+1}) \in \delta^{\Dll_N}((q_n,v_n), w[n])$.
  By Definition~\ref{def:FlN-to-Fl},
  it holds that
  $V_{n+1} = \big(\bigcup_{v_n \in V_n}\delta^{\mc F^\ell_N}((q_n,v_n),w[n])\big)_{\mid G^\ell}$,
  hence $v_{n+1} \in V_{n+1}$.

  \medskip
  \noindent
  \emph{Claim 5:
        If the run $\rho$ enters $\Dll$ at a moment $j\geq m$,
        then $\rho[j{:}]_\projDll$ is an infinite path of $\Dll$ on $w[j{:}]$ starting from an entry state,
        i.e.,
        the word $w$ reaches $\Dll$.}\\
  Proof:
  By Lemma~\ref{lem:dpa-enters-init},
  the run enters $\Dll$ through its entry state.
  Let $\rho[j]_\projDll = (q_j,\Vinit_j)$ be an entry state of $\Dll$ when $\rho$ enters it.
  By Claim~4,
  $\Dll$ has an infinite path $q^\lPlus_j q^\lPlus_{j+1} q^\lPlus_{j+2} \ldots$ on $w[j{:}]$
  starting from $q^\lPlus_j = (q_j,\Vinit_j)$.
  We show that for every $i\geq j$,
  $\rho[i]_\projDll = q^\lPlus_i$.
  For $i=j$,
  we already established that $\rho[i]_\projDll = \rho^\lPlus[i]$.
  If it holds for $i\geq j$,
  then it also holds for $i+1$,
  since the tDPA simulates, when possible, transitions of $\Dll$ on $w$
    \footnote{If the tDPA is in state $q$, and $q_\projDll$ transits into state $\hat q^\lPlus$ of $\Dll$,
              then the tDPA successor state $q'$ satisfies $q'_{\projDll} = \hat q^\lPlus$.
              This is because the tDPA is built from determinized \NFAs built via Definition~\ref{def:FlN},
              and such \NFAs have this property.
              }
  .\label{usage:lem:dpa-DL-transits}
  Thus, $\rho[j{:}]_\projDll$ is an infinite path of $\Dll$ on $w[j{:}]$, and $w$ reaches $\Dll$.

  \medskip
  \noindent
  \emph{Claim~6:
        for all $j \geq m$,
        ${q_j}_\projDll \neq \circ$.}\\
  Proof:
  Let $q_j^\ell = {q_j}_\projDl$.
  By Claim 2,
  $\Dll$ has an entry state for ${q^\ell_j}$.
  Hence, in \textsc{RecBuild($\ell,\Dl$)},
  the guard in Line~\ref{l9} holds for $q^\ell_j$,
  so Line~\ref{l10} maps $\penter^\ell: {q^\ell_j} \mapsto \penter^\lPlus((q^\ell_j,\Vinit_{q^\ell_j}))$.
  We distinguish two cases, $j>0$ and $j=0$.
  Suppose $j>0$.
  Consider the transition $(q_{j-1},x_{j-1},q_j)$ of the tDPA run.
  Let $\ell'\leq \ell$ be the largest number such that
  $({q_{j-1}}_\projDlp,x_{j-1},{q_j}_\projDlp)$ is a transition of the \DFA $\Dlp$,
  where either $\ell'<\ell$ or $\ell = \ell'$ and then $\Dlp$ is $\Dl$ itself.
  In \textsc{RecBuild($\ell',\Dlp$)},
  for the transition $(q^{\ell'}_{j-1},x_{j-1},q^{\ell'}_j)$ of $\Dlp$,
  either (a) Line~\ref{l14} holds true and Line~\ref{l15} is executed, or (b) the guard in Line~\ref{l14} is false.
  In case (a),
  the transition $(q_{j-1},x_{j-1},q_j)$ of the tDPA run is
  $\penter^{\ell'}(q^{\ell'}_{j-1}) \trans{x_{j-1}} \penter^{\ell'}(q^{\ell'}_j)$.
  Note that
  $\penter^{\ell'}(q^{\ell'}_j) = \ldots = \penter^{\ell}(q^{\ell}_j) = \penter^{\ell+1}((q^\ell_j,\Vinit_{q^\ell_j}))$.
  Therefore,
  it holds that
  $q_j = \penter^{\ell+1}((q^\ell_j,\Vinit_{q^\ell_j}))$, so
  ${q_j}_\projDll = (q^\ell_j,\Vinit_{q^\ell_j}) \neq \circ$.
  Now consider case (b).
  Then, $\ell'=\ell$.
  The transition $({q_{j-1}}_\projDl,x_{j-1},{q_j}_\projDl)$ of $\Dl$
  has a refining transition,
  hence ${q_j}_\projDll \neq \circ$.
  Finally, the case of $j=0$ is proven similarly and omitted.

  \medskip
  Finally, we are ready to prove that the word $w$ reaches $\Dll$.
  Claim~6 implies that the run $\rho$ enters $\Dll$ at least once.
  If this happens at a moment $j\geq m$,
  then by Claim~5,
  $w$ reaches $\Dll$.
  Consider the case when $\rho$ enters $\Dll$ for the last time at a moment $n<m$.
  By Lemma~\ref{lem:dpa-enters-init},
  $\rho[n]_\projDll$ is an entry state of $\Dll$.
  Thus,
  $\rho[n{:}]_\projDll$ starts in an entry state,
  and, because no other entering happens,
  for all $j\geq n$ $({q_j}_\projDll,x_j,{q_{j+1}}_\projDll)$ is a transition of $\Dll$.
  Therefore, $\rho[n{:}]_\projDll$ is an infinite path of $\Dll$ on $w[n{:}]$, and $w$ reaches $\Dll$.
\end{proof}

\begin{theorem}\label{thm:ldl-to-dpa}
  For an LDL formula $\varphi$,
  function \textsc{LDLtoDPA} runs in time $2^{2^{O(|\varphi|)}}$\!\! and returns a tDPA with $2^{2^{O(|\varphi|)}}$\!\! states recognizing $L(\varphi)$.
\end{theorem}

\begin{proof}
  To prove that the tDPA recognizes $L(\varphi)$,
  it suffices to show that for every word $w$:
  \li
  \item if $w \in L$, then the deepest \DFA $\Dl$ reached by $w$ has even $\ell$; and
  \item if $w \in \overline{L}$, then the deepest \DFA $\Dl$ reached by $w$ has odd $\ell$.
  \il
  Consider the first case.
  If $w$ reaches some $\Dl$ for odd $\ell$,
  then by Lemma~\ref{lem:ldl-to-dpa},
  it must also reach some child $\Dll$ on even level.
  Hence the deepest level is even.
  The proof of the second case is similar.

  We now analyze the time complexity.
  Lines~\ref{l4} and \ref{l5} have a uniform (independent of $\ell$) time complexity of $\twoexpfi$ (Lemma~\ref{lem:Fl-uniform-bound}).
  Checking if a state has refining states at a deeper level (Line~\ref{l9}) can also be done in time $\twoexpfi$\! by naive enumeration of states in current $\mc P^\ell$ (which has $\twoexpfi$\! states independently of $\ell$);
  the same bound applies to checking for refining transitions.
  Therefore, at each node, \textsc{RecBuild} spends time $\twoexpfi$.

  The recursion depth is bounded by the number of levels in the COCOA for $L(\varphi)$,
  which is $\oneexpfi$,
  yielding an overall running time of $\twoexpfi$\!.
\end{proof}

\section{Translating LDL into LDBA}\label{sec:LDL-to-LDBA}
This section describes a procedure to translate LDL into LDBA.
It can also be applied to weak alternating automata.
We describe the procedure here for completeness, and
in order to exhibit similarities with the procedure of~\cite{DBLP:journals/jacm/EsparzaKS20}.

A transition-based \emph{limit-deterministic B\"uchi automaton} (LDBA) is a tNBA $(\Sigma, Q, q_0, \delta, \alpha)$ having the following structure:
its set of states can be partitioned $Q = \overline{Q_d} \uplus Q_d$
such that
$q_0 \in \overline{Q_d}$,
the transition relation $\delta \cap (Q_d \x \Sigma \x Q_d)$ is deterministic,
there are no transitions from states in $Q_d$ to states in $\overline{Q_d}$, and
$\alpha \subseteq Q_d \x\Sigma\x Q_d$.
Thus, in LDBAs,
every accepting run starts in the nondeterministic part then eventually moves into the deterministic part and satisfies the B\"uchi acceptance condition there.

Before describing our translation LDL-to-LDBA, we state an observation.
Fix an SLTM $\mc M$ and a \DFA wrt.\ $\mc M$.
The \DFA can be viewed as the following LDBA.
The LDBA starts in the SLTM.
The LDBA deterministic partition consists entirely of \DFA's states and transitions,
with all \DFA's transitions declared as accepting in the LDBA.
The LDBA nondeterministically jumps from a state of the SLTM into any state of the \DFA labelled by that SLTM state.
It is not hard to see that the LDBA recognizes the language of the \DFA.

The idea of the translation LDL-to-LDBA is the following.
First, construct \DFAs $(\mc F^1,\ldots,\mc F^k)$ for a given LDL formula $\varphi$,
as described in Section~\ref{subsec:FCW_Al}.
Recall that a word belongs to $L(\varphi)$ if and only if
there exists an even $\ell \in \{0,\ldots,k\}$
such that $\mc F^\ell$ accepts the word but $\mc F^\ellpo$ rejects it,
where, by convention, $\mc F^0$ accepts everything and $\mc F^{k+1}$ accepts nothing.
The LDBA for $L(\varphi)$ starts in the SLTM.
It nondeterministically chooses a moment and an even level $\ell \in \{0,\ldots,k\}$,
then jumps into deterministic component $\mc D^\ell$,
which is a special kind of product $\mc F^\ell \x \overline{\mc F^\ellpo}$
recognizing $L(\mc F^\ell)\cap \overline{L(\mc F^\ellpo)}$.
For a run in $\mc F^\ell \x \overline{\mc F^\ellpo}$ to be accepting,
it has to always stay in $\mc F^\ell$ and infinitely often leave $\mc F^\ellpo$,
which corresponds to the B\"uchi acceptance condition.

Figure~\ref{fig:alg:ldl-to-ldba} describes the top-level function \textsc{LDLtoLDBA}.
The function returns an LDBA with states $S\uplus Q_d$,
where $S$ are the states of the SLTM and $Q_d$ are the states of all products
$\mc F^\ell \x \overline{\mc F^\ellpo}$ for even levels $\ell$.
The LDBA initial state is the SLTM initial state $s_0$.
The transitions $\delta^S\uplus\delta_\mathit{jump}\uplus\delta_d$ are the union of those of the SLTM,
of the products, and nondeterministic transitions from the SLTM to the products.
The choice of whether to stay in the SLTM or to jump into a product component
and the choice of the product are the only nondeterministic choices in the LDBA
(there is no nondeterminism in the SLTM).
These nondeterministic transitions are defined in lines~\ref{ldl-to-ldba:line:start-jump}\,--\,\ref{ldl-to-ldba:line:end-jump}:
the SLTM can jump into any state of any product provided it is labelled by that SLTM state.
The jumping uses $\epsilon$-transitions that can be removed in a standard way.

Figure~\ref{fig:alg:ldl-to-ldba:product} describes the product $\mc F^\ell\x\overline{\mc F^\ellpo}$ construction.
The product has states $Q_d^\ell \subseteq Q^{\Fl}\uplus Q^{\mc F^\ellpo}$:
it contains all states of $Q^{\Fellpo}$;
it contains a state $q^{\Fl}$ of $\Fl$
if and only if
$\Fellpo$ has no state projecting onto $q^{\Fl}$
(lines~\ref{ldl-to-ldba:product:start-states}\,--\,\ref{ldl-to-ldba:product:end-states}).
Line~\ref{ldl-to-ldba:product:Fellpo-transitions} adds all transitions of $\Fellpo$ to the product.
Lines~\ref{ldl-to-ldba:product:start-complementing-Fellpo}\,--\,\ref{ldl-to-ldba:product:end-complementing-Fellpo} add the transitions of $\Fl$ that are not possible in $\Fellpo$:
each time $\Fellpo$ cannot make a transition while $\Fl$ transits into a state $q^{\Fl}\!$,
the product either transits into an entry state of $\Fellpo$ for $q^{\Fl}$ if such exists,
or otherwise transits into $q^{\Fl}$.
Such transitions are set as accepting.

\begin{figure}
  \begin{algorithmic}[1]
    \algrenewtext{Function}[2]{\algorithmicfunction\ \textproc{#1}(#2):}
    \algrenewcommand\algorithmicthen{:}
    \algrenewcommand\algorithmicdo{:}
    \algrenewcommand\algorithmicelse{\textbf{else:}}
    \Function{LDLtoLDBA}{$\varphi$}
      \State translate $\varphi$ into a weak alternating automaton $\mc A$
      \State construct an SLTM $(\Sigma,S,s_0,\delta^S)$ for $L(\mc A)$
      \State derive \DFA $\mc F^0 = (\Sigma,S,\delta^S\!,\,\mathit{identity})$ from the SLTM
      \State construct \DFAs $\mc F^1,\ldots,\mc F^k$ with entry states
      \State let $\mc F^{k+1}$ be the \DFA with no states
      \State initialize $Q_d,\delta_d,\alpha_d,\delta_\mathit{jump}$ as empty
      \For{every even $\ell$ in $\{0,\ldots,k\}$}
        \State $Q_d^\ell,\delta_d^\ell,\alpha_d^\ell = \textsc{Product}(\mc F^\ell,\mc F^\ellpo)$
        \State add $(Q_d^\ell,\delta_d^\ell,\alpha_d^\ell)$ to $(Q_d,\delta_d,\alpha_d)$
      \EndFor
      \For{every state $s$ of the SLTM}\label{ldl-to-ldba:line:start-jump}
        \For{every state $q_d$ in $Q_d$ labelled by $s$}
          \State add epsilon transition $(s,\epsilon,q_d)$ to $\delta_\mathit{jump}$
        \EndFor
      \EndFor\label{ldl-to-ldba:line:end-jump}
      \State \Return $(\Sigma,S\uplus Q_d,s_0,\delta^S\uplus\delta_\mathit{jump}\uplus\delta_d,\alpha_d)$
    \EndFunction
  \end{algorithmic}
  \caption{Translating LDL into LDBA.}
  \label{fig:alg:ldl-to-ldba}
\end{figure}

\begin{figure}
  \begin{algorithmic}[1]
    \algrenewtext{Function}[2]{\algorithmicfunction\ \textproc{#1}(#2):}
    \algrenewcommand\algorithmicthen{:}
    \algrenewcommand\algorithmicdo{:}
    \algrenewcommand\algorithmicelse{\textbf{else:}}
    \Function{Product}{$\Fl\!, \mc F^\ellpo$}
      \State \textbf{if} $\mc F^\ellpo$ has no states: \Return{$Q^{\Fl}\!\!,\,\delta^{\Fl}\!\!,\,\delta^{\Fl}$}
      \State initialize $Q^\ell_d,\delta^\ell_d,\alpha^\ell_d$ as empty
      \For{each state $q^{\Fl}$ that has no entry state in $\mc F^\ellpo$}\label{ldl-to-ldba:product:start-states}
        \State add $q^{\Fl}$ to $Q^\ell_d$
      \EndFor
      \State add $(Q^{\mc F^\ellpo}\!\!\,,\delta^{\mc F^\ellpo})$ to $(Q^\ell_d,\delta^\ell_d)$\label{ldl-to-ldba:product:end-states}\label{ldl-to-ldba:product:Fellpo-transitions}

      \While{$\exists q \in Q^\ell_d$, $x\in\Sigma$, $\tilde q^{\Fl}$: $\Fl$ has $q_{\mid \Fl}\trans{x}\tilde q^{\Fl}$ but $\delta^\ell_d$ has no transition for $q,x$}\label{ldl-to-ldba:product:start-complementing-Fellpo}
        \State let $\tilde q$ be an entry state $\tilde q^{\mc F^\ellpo}_e$\,for $\tilde q^{\Fl}$\,if it exists, otherwise $\tilde q^{\Fl}$
        \State add $(q, x,\tilde q)$ to $\delta^\ell_d$ and $\alpha^\ell_d$
      \EndWhile\label{ldl-to-ldba:product:end-complementing-Fellpo}

      \State \Return{$Q^\ell_d,\delta^\ell_d,\alpha^\ell_d$}
    \EndFunction
  \end{algorithmic}
  \caption{Constructing the product $\mc F^\ell \x \overline{\mc F^\ellpo}$
           to recognize the language $L(\Fl)\cap\overline{L(\mc F^\ellpo)}$.}
  \label{fig:alg:ldl-to-ldba:product}
\end{figure}

\begin{theorem}\label{LDL-to-LDBA}
  Given an LDL formula $\varphi$,
  the algorithm in Figures~\ref{fig:alg:ldl-to-ldba}\hspace{0.3mm}--\ref{fig:alg:ldl-to-ldba:product} constructs a transition-based LDBA recognizing $L(\varphi)$
  with $2^{2^{O(|\varphi|)}}$states in time $2^{2^{O(|\varphi|)}}$.
\end{theorem}

\begin{proof}
  Fix an LDL formula $\varphi$.

  Suppose a word $w = x_0 x_1 \ldots$ satisfies $\varphi$.
  There exists an even number $\ell \in \{0,\ldots,k\}$
  such that $w \in L(\Fl)\cap \overline{L(\Fellpo)}$.
  Since $\Fl$ accepts $w$,
  let $m$ be a witnessing moment and let $q^{\Fl}_m q^{\Fl}_{m+1}\ldots$ be an infinite run of $\Fl$,
  where the state $q^{\Fl}_m$ is labelled with the SLTM state $s = \delta^S(s_0, x_0 \ldots x_{m-1})$.
  We describe a run $\rho$ of the LDBA on $w$ and show it is accepting.
  It stays in the SLTM until moment $m$,
  then jumps from the SLTM state $s$ into state $q$ of the product $\Fl\x\overline{\Fellpo}$,
  where $q$ is an entry state of $\Fellpo$ for $q^{\Fl}_m$ if it exists, otherwise $q^{\Fl}_m$.
  Afterwards, the run $\rho$ evolves deterministically.
  Since the word $w$ is not accepted by $\Fellpo$,
  any run of $\Fellpo$ on $w$ starting at any moment $m'$ eventually leaves $\Fellpo$,
  or we cannot even enter $\Fellpo$ at moment $m'$.
  Thus,
  there exists infinitely many moments $m' \geq m$,
  where the LDBA run either takes a transition of $\Fl$ or enters $\Fellpo$,
  which corresponds to an accepting transition
  (Lines~\ref{ldl-to-ldba:product:start-complementing-Fellpo}\hspace{0mm}--\hspace{0.2mm}\ref{ldl-to-ldba:product:end-complementing-Fellpo} in Figure~\ref{fig:alg:ldl-to-ldba:product}).
  Hence, the LDBA run $\rho$ is accepting.

  \newcommand\ql{q^\ell}
  \newcommand\qellpo{q^{\ellpo}}
  Suppose a word $w = x_0 x_1 \ldots$ does not belong to $L(\varphi)$.
  We show that every LDBA run on $w$ is non-accepting.
  \li
  \- Fix an LDBA run $\rho$ on $w$ that leaves the SLTM component
     at some moment $m$ and enters $\Fl\x\overline{\Fellpo}$,
     for some even $\ell\in\{0,\ldots,k\}$.
     (The special case of $k$ being even and $\ell=k$ is not possible as it contradicts $w \not\in L(\varphi)$.)
  \- The infinite sequence $q^\ell_m q^\ell_{m+1}\ldots$,
     where $q^\ell_i = \rho[i]_{\mid \Fl}$ for all $i\geq m$,
     is an infinite run of $\Fl$ starting at moment $m$,
     so the word is accepted by $\Fl$.
  \- We are going to show that the LDBA run $\rho$ eventually enters $\delta^{\Fellpo}$ then never leaves it,
     meaning that the run is non-accepting.
     The argument is similar to the proof of Lemma~\ref{lem:ldl-to-dpa}.
  \- Since $w \not\in L(\varphi)$,
     $w$ is accepted by the obligation graph $\bar {\mc G}$ for $L(\neg\varphi)$.
     Let  $v_0 \trans{x_0} v_1 \trans{x_1} \ldots$ be an accepting run of $\bar {\mc G}$ on $w$.
  \- The sequence $(q^\ell_m,v_m)\trans{x_m}(q^\ell_{m+1},v_{m+1})\trans{x_{m+1}} \ldots$ ends in an SCC of $\mc F^\ellpo_N$,
     where $\mc F^\ellpo_N$ is the \NFA constructed by Definition~\ref{def:FlN}.
     That SCC is not removed in Step 2 of the definition
     because the sequence $v_m \trans{x_m} v_{m+1} \trans{x_{m+1}}\ldots$ visits a $\bar {\mc G}$-accepting transition infinitely often.
     Let $M$ be the moment after which the sequence stays in that SCC.
  \- We next show that ($\dagger$) for every $j\geq M$,
     the entry state $q^\ellpo_e$ for $q^\ell_j$ exists and starts an infinite path in $\Fellpo$ on $w[j{:}] = x_j x_{j+1}\ldots$.
     This implies that the LDBA run either eventually enters $\delta^{\Fellpo}$ and then never leaves it, or
     the LDBA run does not \emph{enter} $\delta^{\Fellpo}$ at moments $j\geq M$ because it already follows it.
     In both cases, the LDBA run is non-accepting.
  \- We now prove the claim ($\dagger$).
     Since $(q^\ell_j,v_j)$ belongs to an SCC of $\Fellpo_N$\!,
     by Definition~\ref{def:FlN-to-Fl}
     the \DFA has an entry state of the form $q^\ellpo_e = (q^\ell_j,\Vinit_{q_j^\ell})$
     where $v_j \in \Vinit_{q_j^\ell}$.
     Starting the subset construction from $(q^\ell_j,\Vinit_{q_j^\ell})$,
     let $(q_i^\ell,V_i)$ be the state reached at moment $i\geq j$.
     An induction on $i$ shows that $v_i\in V_i$:
     this holds initially, and the transition
     $(q_i^\ell,v_i)\trans{x_i}(q_{i+1}^\ell,v_{i+1})$ of $\mc F^{\ellpo}_N$
     implies that $v_{i+1}$ belongs to the successor subset $V_{i+1}$.
     Hence, none of these subsets is empty,
     so $q^\ellpo_e$ indeed starts an infinite path of $\mc F^{\ellpo}$ on $w[j{:}]$.
  \il

  Finally,
  the complexity claim follows from
  the fact that the SLTM has $2^{2^{O(|\varphi|)}}$\,states and constructible within the same time bound (Fact~\ref{fact:SLTM}),
  there are $2^{O(|\varphi|)}$ levels (Fact~\ref{fact:cocoa-nof-levels}), and
  each level \DFA has $2^{2^{O(|\varphi|)}}$\,states and constructible within the same time bound (Lemma~\ref{lem:Fl-size}).
\end{proof}

\section{Related Works}\label{sec:related-works}

Conceptually,
the asymptotically optimal translations of LTL into deterministic $\omega$-automata
follow two broad routes.

The first route proceeds by first translating an LTL formula into an NBA~\cite{VardiWolper-AnAutomataTheoretic}
and then determinizing the NBA using either the Safra construction~\cite{DBLP:conf/focs/Safra88,Pit07}
or the Muller--Schupp construction~\cite{MULLER199569,10.1007/978-3-540-70575-8_59}.
These constructions encode, in a bounded-size deterministic state,
sufficient information about the set of runs of the NBA and about their acceptance progress.
In the Safra case,
a deterministic state is encoded using so-called history trees.
In the Muller--Schupp case,
a deterministic state is encoded using levels of so-called split trees.
L\"oding and Pirogov~\cite{loding_et_al:LIPIcs.ICALP.2019.120}
unify the Safra-style and Muller--Schupp-style determinization approaches by
giving a single meta-construction from NBAs to DPAs.
Their construction uses ranked slices as deterministic states,
which is an ordered tuple of pairwise disjoint sets of NBA states, annotated with ranks.
This data structure incorporates history trees used by Safra and split trees used by Muller and Schupp.
The transition function in the unified construction
can be instantiated to obtain either the Safra construction or the Muller--Schupp construction.

A second route,
introduced by Esparza, Kret\'insk\'y, and Sickert~\cite{DBLP:journals/jacm/EsparzaKS20},
translates LTL directly into DRA.
The crucial step is decomposition of LTL into a Boolean combination of simpler languages,
that are later translated separately into DRA;
the final DRA is the disjunction of the individual DRAs.
The authors use LTL formulas as automata states,
which evolve according to how LTL formulas transform after reading a letter.

Our approach is closer to that of Esparza et al.\ 
because it also relies on a language decomposition,
namely the decomposition via natural colors~\cite{DBLP:conf/fsttcs/EhlersS22}.
In the following sections,
we first describe their approach then compare it with ours in detail.

One final note.
In the paper~\cite{10.1145/2362355.2362357},
Boker and Kupferman study the problem of constructing an NCA
for a co-B\"uchi language given by an NBA or another automaton model
with the expressive power of omega-regular languages.
Given an NBA,
they introduce the augmented subset construction SC(NBA)$\x$NBA
which builds the product of the NBA with the subset construction SC(NBA) applied to it.
The product construction SLTM$\x$NBA used in our Definition~\ref{def:FlN}
of $\mc F^\ell_N$ for the case $\ell=1$ resembles the augmented subset construction SC(NBA)$\x$NBA.
Moreover,
the acceptance condition in the definition of $\mc F^\ell_N$
coincides with the one they place on top of the augmented subset construction when constructing the NCA.

\paragraph*{The work of Esparza, Kret\'insk\'y, and Sickert~\cite{DBLP:journals/jacm/EsparzaKS20}}

Esparza et al. propose a decomposition
$L(\varphi) = \bigcup_i U_i \cap V_i$
of the language of a given LTL formula $\varphi$
into a finite union of intersections of a deterministic B\"uchi language $U_i$
and a co-B\"uchi language $V_i$.
They construct a DRA recognizing each language $U_i \cap V_i$,
and the final DRA is obtained by taking the union of these DRAs.
To understand the decomposition result,
we need two notions:
the after-formula automaton and rewriting a formula using advice.

Intuitively,
the after-formula automaton tracks,
after reading a finite prefix,
which formula the remaining suffix has to satisfy for the whole word to satisfy the original formula.
Each state of such an automaton is an LTL formula that is a Boolean combination of subformulas of the original formula.
Since there are only doubly exponentially many such formulas,
the after-formula automaton has doubly exponential size.
The after-formula automaton resembles the right-congruence automaton used in our work:
in fact, their minimal versions are isomorphic.
Given a prefix $w_{0i}$ and an LTL formula $\varphi$,
let $\mathit{aft}(\varphi,w_{0i})$ denote the LTL formula that the suffix has to satisfy
for the whole word to satisfy $\varphi$.
For example,
$\mathit{aft}(\LTLX a,x) = a$ for any letter $x$,
and $\mathit{aft}(a \LTLU b,a) = a \LTLU b$.

We now explain the concept of rewriting a formula using advice.
Advice consists of two parts:
the assumption that exactly the formulas in a set $M$ of co-safety temporal subformulas of $\varphi$
hold infinitely often on the word $w$,
and the assumption that exactly the formulas in a set $N$ of safety temporal subformulas of $\varphi$
eventually always hold on $w$.
If a word satisfies the advice assumptions,
then the original formula can be rewritten so that the word satisfies the rewritten formula
if and only if it satisfies the original formula.
For instance,
given the LTL formula $\varphi = \LTLG(b \LTLU c)$ and
assuming that the word satisfies the advice $M = \{b \LTLU c\}$,
checking whether the word satisfies $\varphi$ amounts to checking
whether it satisfies the rewritten formula
$\varphi[M]_\nu = \LTLG(b \mathcal{W} c)$.
We denote by $\varphi[M]_\nu$ the formula rewritten using the advice $M$,
and by $\varphi[N]_\mu$ the formula rewritten using the advice $N$.
The important insight is that rewriting using advice $M$ produces a safety formula
(indicated by the subscript $_\nu$),
whereas rewriting using advice $N$ produces a co-safety formula
(indicated by the subscript $_\mu$).
This enables the language decomposition into B\"uchi/co-B\"uchi intersections.

We are now ready to state how Esparza et al. obtain the desired decomposition.
To this end, they prove the Master Theorem:
a word $w = x_0 x_1 \ldots$ satisfies $\varphi$
if and only if
there exist
an advice set $M$ of co-safety temporal subformulas of $\varphi$ and
an advice set $N$ of safety temporal subformulas of $\varphi$
such that:
\lo
\-[(1)] ~$\exists i{:}\ w_i \models \mathit{aft}(\varphi,w_{0i})[M]_\nu$, where $w_{0i} = x_0 \ldots x_{i-1}$ and $w_i = x_i x_{i+1}\ldots$;
\-[(2)] ~$w \models \bigwedge_{\psi \in M} \LTLGF \psi[N]_\mu$; and
\-[(3)] ~$w \models \bigwedge_{\psi \in N} \,\LTLFG \psi[M]_\nu$.
\ol
They show that conditions~(1) and~(3) can be expressed using deterministic co-B\"uchi automata,
and condition~(2) using a deterministic B\"uchi automaton;
therefore, their intersection is a DRA.
The number of such individual DRAs is at most singly exponential,
which is the maximal number of advice-set pairs $(M,N)$.
Finally, the DRA recognizing $L(\varphi)$ is obtained as the disjunction of the individual DRAs.

\paragraph*{Comparison with our work}

\parbf{Input.}
Our procedure accepts weak alternating B\"uchi automata as input;
therefore, it can also translate LDL~\cite{DBLP:conf/ijcai/GiacomoV13,DBLP:journals/iandc/FaymonvilleZ17},
which is more expressive than LTL and covers all omega-regular languages,
with optimal asymptotic complexity.
Esparza et al.'s procedure relies on syntactic rewriting of LTL formulas,
so extending it to LDL would require additional work.
The procedure of Esparza et al.\ was extended to LTL with past (pLTL) in~\cite{PitLTLpast-to-EJK}.
Formulas of pLTL can be translated into NBAs of singly exponential size~\cite{DBLP:conf/mfcs/GastinO03}.
Since our procedure internally works with the SLTM and NBAs for the formula and its negation,
extending our procedure to pLTL would require a way to construct SLTMs from pLTL formulas.

\parbf{Output.}
Our procedure outputs tDPAs, whereas Esparza et al.'s procedure outputs DRAs.
In general,
translating a DRA into a DPA may incur an exponential blow-up~\cite{646837.708484},
so DRAs can be more succinct.
However, the blow-up does not occur for the DRAs obtained from LTL
because the number of Rabin pairs is logarithmic in the number of states.
Our procedure also directly constructs COCOA;
Esparza et al.'s work can be used to construct COCOA via the route
LTL$\to$DRA$\to$DPA$\to$COCOA.
Both approaches can be used to construct LDBAs of doubly exponential size,
which is tight for full LTL.
This bound is not optimal for the fragment LTL$_{\setminus GU}$,
i.e., formulas in which no $\LTLU$ operator occurs in the scope of a $\LTLG$ operator;
for this fragment, singly exponential LDBAs can be constructed~\cite{10.1007/978-3-662-46681-0_57}.
There are formulas~\cite{KV98,AL04} in the fragment LTL$_{\setminus GU}$
for which the smallest deterministic automaton is at least doubly exponential.
Therefore, for such formulas,
the after-formula automaton and the SLTM are also at least doubly exponential.
Thus, for these formulas,
both our construction and Esparza et al.'s construction produce LDBAs of non-optimal doubly exponential size.

\parbf{Decomposition.}
Both works decompose the language $L(\varphi)$ of a given LTL formula $\varphi$
into a combination of simpler languages.
In this paper,
the language is decomposed into a chain of co-B\"uchi languages
$L^1 \supset \dots \supset L^k$.
This COCOA decomposition is canonical,
whereas the decomposition provided by the Master Theorem is not.

\parbf{Further comparison notes.}
\li
\- The minimized after-formula automata used by Esparza et al. and
   the right-congruence automata used in this work are isomorphic.
\- Our translation proceeds incrementally by analyzing $L^1$ through $L^k$,
   whereas the Esparza et al.\ procedure enumerates independent choices of advice sets $M$ and $N$
   and combines the resulting automata by disjunction.
\- The correctness proofs of both approaches rely on a ``delaying'' property.
   In our work, floating automata are insensitive to delays:
   staying in the SLTM for an arbitrary finite number of steps before jumping into the floating-automaton component
   does not affect whether the word is accepted.
   In Esparza et al.'s work,
   the Master Theorem still holds if the index $i$ in condition~(1) is required to be larger than any given number,
   while conditions (2) and (3) can take arbitrary suffix words instead of the original one.
\- When translating into LDBA,
   both procedures use the SLTM---equivalently, the minimized after-formula automaton---as the initial component of the LDBA
   before jumping into one of exponentially many deterministic accepting components.
\il

\section{Conclusion}\label{sec:conclusion}
In this paper,
we provided a direct translation procedure
from weak alternating automata (WAA)
to a chain-of-co-B\"chi automata representation (COCOA),
and the variants translating into deterministic parity automata and limit-deterministic B\"uchi automata.
All procedures have optimal asymptotic complexity.
Together with the encoding of linear dynamic logic (LDL) into weak alternating automata,
the procedure provides a novel way of translating from LDL into aforementioned automata.

We looked at the translation conceptually and avoided discussing implementation details.
We now make a few remarks regarding the efficiency of the translation procedure in practice.
First, the disjoint SCCs of the SLTM of a given language naturally induce parallelization,
because such SCCs can be processed independently of each other.
Second, the algorithms for constructing a new COCOA level from a previous level are standard graph algorithms having efficient implementations.
Third, constructing the SLTM itself is also naturally parallelizable.
One challenge is to implement the SLTM construction efficiently,
since a naive implementation, as described in Section~\ref{sec:SLTM},
may perform an unnecessarily large number of equivalence checks between pairs of weak alternating automata,
with each check requiring Miyano-Hayashi's construction.
When translating from LTL,
the minimized after-formula automaton from~\cite{DBLP:journals/jacm/EsparzaKS20},
which is constructed using LTL formula rewriting and satisfiability checking,
can be used instead of the SLTM.
We leave further efficiency analysis to future work.

We conclude with a remark on the problem of direct construction of HD-tNCAs
that avoids constructing DCAs first.
Despite being exponentially more succinct than tDCAs,
there are currently no procedures that
guarantee to construct HD-tNCAs
while always avoiding the intermediate blow-up of the size of the equivalent tDCA \emph{and} that take full LTL as input.
Consider the language family $L_1,L_2,\ldots$
that Kuperberg and Skrzypczak define
to prove the exponential succinctness gap between HD-tNCAs and tDCAs~\cite[Thm.1]{kuperberg2015determinisation}.
Each language $L_n$ can be expressed using HD-tNCA $\mc C_n$ of size $2n+1$,
and a simple modification of $\mc C_n$ yields WAA $\mc B_n$ of the same size
(in fact, \emph{nondeterministic} weak B\"uchi automaton).
They show that any tDCA recognizing $L(\mc C_n)$ (hence $\mc B_n$) must have at least $\frac{2^n}{2n+1}$ states.
We can run $\mc B_n$ through our translation procedure to produce HD-tNCA~%
  \footnote{Define the language of the HD-tNCA as
            $\{w \mid \forall J.\exists \text{SLII}^{L(\mc B_n)}_{J,w} w': w' \in L(\mc B_n)\}$,
            then run Definition~\ref{def:FlN} for $\ell=2$ assuming that
            level-1 \DFA has the universal language; use NBA $\mc B_n$ as the obligation graph.}.
The SLTM for the language of $\mc B_n$ has a single state.
Applying Definition~\ref{def:FlN} yields an \NFA that is actually deterministic (i.e., \DFA),
so applying the HDification \DFA-to-HD-tNCA used in the proof of Lemma~\ref{lem:DFA-to-HD-tNCA}
gives HD-tNCA $\mc C'_n$ which has only $2n$ states.
So, for this particular example, our procedure WAA$\to$HD-tNCA avoids the exponential blow-up, right?
Not quite.
The above reasoning hid the construction of the single-state SLTM and assumed it was given.
Running the algorithm to construct the SLTM under the hood executes the Miyano-Hayashi construction,
which blows-up exponentially on $\mc B_n$.
Thus, if we need to construct the SLTM,
our translation still exhibits the unnecessary exponential blow-up.
(The alternative path,
 by first translating NBA $\mc B_n$ into NCA using~\cite{10.1145/2362355.2362357},
 then running the $k$-breakpoint construction of \cite{kuperberg2018width}
 also results in the exponential blow-up.)
In \cite{10.1007/978-3-030-24886-4_14},
the authors,
first,
suggest the logic E$\nu$TL (``eventually safety'') that
can express the family $L_1,L_2,\ldots$ using formulas of linear size only,
and second, the algorithm to construct history-deterministic automata.
That algorithm is doubly exponential in the worst case,
but it performs in linear time on the language family $L_1,L_2,\ldots$.
The formulas of E$\nu$TL always have SLTMs with a single state,
which avoids the need to compute it that hindered our translation in the previous example.
The thorough analysis how our translation performs on E$\nu$TL is left for future work.

\medskip
\noindent
{\small \textbf{Acknowledgments}. We thank anonymous LICS reviewers for numerous constructive suggestions.}

\bibliography{bib.bib}

\end{document}